\documentclass[10pt]{elsarticle}
\usepackage{graphicx}
\usepackage{amssymb,amsmath}
\usepackage{amsmath}
\usepackage{mathtools}
\usepackage[margin=1in,papersize={8.5in,11in}]{geometry}
\usepackage{times}
\usepackage[utf8]{inputenc}
\usepackage[english]{babel}
\usepackage{amsthm}
\usepackage{float}
\usepackage{algorithm}
\usepackage{algpseudocode}
\usepackage{amsfonts,mathrsfs}
\usepackage{bm}
\usepackage{bbm}
\usepackage{epstopdf}
\usepackage{subcaption}
\graphicspath{ {images/} }
\usepackage{comment}
\usepackage[table]{xcolor}

\newtheorem{theorem}{Theorem}[section]
\newtheorem{corollary}{Corollary}[theorem]
\newtheorem{lemma}[theorem]{Lemma}

\makeatletter
\newcommand*{\rom}[1]{\expandafter\@slowromancap\romannumeral #1@}
\makeatother

\def\L{\mathcal{L}} 
\def\C{\mathcal{C}} 
 
\def\P{\mathcal{P}}
\def\Q{\mathcal{Q}}
\def\F{\mathcal{F}} 
\def\I{\mathcal{I}}

\def\Pi{\bm {P}}

\def\Ai{\bm{A}}
\def\Mi{\bm{M}}

\begin{document}

\begin{frontmatter}

\title{Faber Approximation to the Mori-Zwanzig Equation}

\author[ucsc]{Yuanran Zhu}
\author[ucsc]{Daniele Venturi\corref{correspondingAuthor}}
\cortext[correspondingAuthor]{Corresponding author}
\ead{venturi@ucsc.edu}

\address[ucsc]{Department of Applied Mathematics and Statistics\\ University of California, Santa Cruz}
\journal{ArXiv}

\begin{abstract}
We develop a new effective approximation of the Mori-Zwanzig equation based on operator series expansions of the orthogonal dynamics propagator.  In particular, we study the Faber series, which yields asymptotically optimal approximations converging at least $R$-superlinearly with the polynomial order for linear dynamical systems. We provide a through theoretical analysis of the new method 
and present numerical applications to random wave propagation and harmonic chains of oscillators interacting on the Bethe lattice and on graphs with arbitrary topology.
\end{abstract}

\end{frontmatter}

\section{Introduction}

The Mori-Zwanzig (MZ) formulation is a technique from irreversible 
statistical mechanics that allows us to develop formally exact 
evolution equations for quantities of interest (phase space functions) 
in nonlinear dynamical systems. One of the main advantages 
of developing such exact equations is that they provide a 
theoretical starting point to avoid integrating 
the full dynamical system and solve directly for the 
quantities of interest, thus reducing the computational 
cost significantly.
As an example, consider a large system of interacting particles,
and suppose we are interested in studying the motion of one 
specific particle. By applying the MZ formulation to the 
equations of motion of the full particle 
system, it is possible extract a formally exact 
generalized Langevin equation (MZ equation) governing 
the position and the momentum of the particle of interest.
This is at the basis of microscopic physical theories of 
Brownian motion \cite{VanKampenOppenheim,Chaturvedi}.
Computing the solution to the MZ equation 
is a very challenging task that relies on 
approximations and appropriate numerical schemes. 
One of the main difficulties is the approximation of the 
memory integral (convolution term), which 
encodes the effects of the so-called orthogonal 
dynamics in the observable of interest. 
The orthogonal dynamics is essentially a 
high-dimensional flow that satisfies a 
complex integro-differential equation. 
In statistical systems far from equilibrium, such flow 
has the same order of magnitude and 
dynamical properties as the observable 
of interest, i.e., there is no scale separation between the 
observable of interest and the orthgonal dynamics. 
In these cases, the computation of the MZ memory 
can be addressed only by problem-class-dependent 
approximations.
The first effective technique to approximate 
the MZ memory integral was developed by H. Mori 
in \cite{mori1965continued}. The method relies on 
on continued fraction expansions, and it can be conveniently 
formulated in terms of recurrence relations
\cite{Snook,Karasudani,lee1982solutions,lee1982eq,florencio1985exact}. 
The continued fraction expansion method of Mori
made it possible to compute the exact solution 
to important prototype problems in statistical mechanics, 
such as the dynamics of the auto-correlation function of a tagged 
oscillator in an harmonic chain \cite{espanol1996dissipative,kim2000dynamics}.
Other effective approaches to approximate 
the MZ memory integral rely on perturbation methods 
\cite{watts1977perturbation,singwi1968theory,venturi2014convolutionless},  
mode coupling techniques, 
\cite{alder1970decay,sjogren1979kinetic,sjogren1980numerical}, 
or functional approximation methods \cite{Fox,harp1970time, Moss1}. 
In a parallel effort, the applied mathematics 
community has, in recent years, attempted to derive general 
easy-to-compute representations of the MZ memory 
integral \cite{VenturiBook,Parish,Gouasmi}. 
In particular, various approximations 
such as the $t$-model \cite{Chorin,Chorin1,Stinis,Chertock2008Modified}, 
hierarchical perturbation methods 
\cite{stinis2015renormalized,Yuan1,venturi2014convolutionless}, 
and data-driven methods \cite{Lei} were proposed 
to address approximation of the MZ memory 
integral in situations where there 
is no clear separation of scales between 
the resolved and the unresolved dynamics. 

In this paper, we study a new approximation of the MZ 
equation based on {\em global} operator series 
expansions of the orthogonal dynamics propagator. 
In particular, we study the Faber series, 
which yields asymptotically optimal approximations
converging at least $R$-superlinearly with 
the polynomial order.
The advantages of expanding the orthogonal 
dynamics propagator in terms of globally 
defined operator series are similar to those we obtain when we approximate 
a smooth function in terms of orthogonal  
polynomials rather than Taylor series \cite{Hesthaven}. 
As we will see,  the proposed MZ memory approximation 
method based on global operator series 
outperform in terms of accuracy and computational 
efficiency the hierarchical memory approximation  
techniques discussed in \cite{Yuan1,stinis2007higher}, 
which are based on Taylor-type expansions. 

This paper is organized as follows. In Section \ref{sec:review}, 
we briefly review the MZ formulation, and discuss 
common choices of projection operators. 
In Section \ref{sec:New_expansion} we develop new 
series expansions of the MZ memory integral based on 
operator series of the orthogonal dynamics propagator.  
We also develop exact MZ equations for 
the mean and the auto-correlation function of an observable 
of interest, and determine their analytical solution 
through Laplace transforms. In Section \ref{sec:convergence} we 
perform a thorough convergence analysis of the memory 
approximation methods we propose in this paper.
In Section \ref{sec:application} we demonstrate 
the accuracy and effectiveness of the Faber 
approximation of the Mori-Zwazig equation. 
Specifically, we study two-dimensional random waves 
in an annulus, and the velocity auto-correlation 
function of a tagged oscillator in harmonic chains interacting 
on the Bethe lattice and on graphs with arbitrary topology.

\section{The Mori-Zwanzig Formulation}
\label{sec:review}
Consider the following nonlinear dynamical system evolving on a smooth manifold 
$\Gamma\subseteq \mathbb{R}^N$
\begin{equation}
\frac{d\bm x}{dt} = \bm F(\bm x),\qquad \bm x(0)=\bm x_0.
\label{eqn:nonautonODE}%
\end{equation}
For simplicity, here we assume that $\Gamma=\mathbb{R}^N$.  
The dynamics of any scalar-valued phase space function $u:\Gamma\to 
\mathbb{R}$ (quantity of interest) can be expressed in terms of a 
semi-group of operators acting on the space of observables , i.e., 
\begin{equation}
 u(\bm x(t)) =e^{(t-s)\L} u (\bm x(s)) ,\quad \textrm{where}\quad \L u(\bm x)= \bm F(\bm x)\cdot \nabla u(\bm x).
\label{Koopman}
\end{equation}
The operator $e^{(t-s)\L}$ is known as  Koopman 
operator \cite{Koopman1931}.
The subspace of functions of $u(\bm x)$ 
can be described conveniently by means of a projection 
operator $\P$, which selects from an arbitrary 
function $f(\bm x)$ the part $\P f(\bm x)$ 
which depends on $\bm x$ only through $u(\bm x)$.
The nature, mathematical properties and connections 
between $\P$ and the observable $u$ are discussed 
in detail in \cite{Dominy2017}.  
For now, it suffices to assume  that $\P$ is 
a bounded linear operator, and that $\P^{2} = \P$.
Also, we denote by  $\Q=\I-\P$ the complementary projection, 
$\I$ being the identity operator.
The MZ formalism describes the evolution of 
observables initially in the image of $\P$. 
Because the evolution of observables 
is governed by the semi-group $e^{t\L}$, 
we seek an evolution equation for $e^{t\L}\P$. By using 
the well-known Dyson identity 
\begin{align}
e^{t\L}=e^{t\Q\L}+
\int_0^t e^{s\L}\P\L e^{(t-s)\Q\L}ds,
\label{Dyson}
\end{align}
we obtain 
\begin{align}
\frac{d}{dt} e^{t\L}\P = e^{t\L}\P\L \P + 
e^{t\Q\L}\Q\L\P+ \int_0^t e^{s\L}\P\L e^{(t-s)\Q\L}\Q\L\P ds.
\label{MZKoop}
\end{align}
By applying this equation to an element $u_0=u(\bm x(0))$ in the 
image of $\P$, we obtain the well-known MZ equation
\begin{align}
\frac{\partial}{\partial t}e^{t\mathcal{L}}u_{0}
&=e^{t\mathcal{L}}\mathcal{PL}u_{0}
+e^{t\mathcal{QL}}\mathcal{QL}u_{0}+\int_0^te^{s\mathcal{L}}\mathcal{PL}
e^{(t-s)\mathcal{QL}}\mathcal{QL}u_{0}ds.
\label{equ:MZ_total}
\end{align}
We emphasize that equation \eqref{equ:MZ_total} 
is completely equivalent to \eqref{Koopman}. 
Acting on the left with $\P$, yields the evolution 
equation for projected dynamics
\begin{align}\label{reduced order equation}
\frac{\partial}{\partial t}\mathcal{P}e^{t\mathcal{L}}u_{0}
=\mathcal{P}e^{t\mathcal{L}}\mathcal{PL}u_{0}
+\int_0^t\mathcal{P}e^{s\mathcal{L}}\mathcal{PL}
e^{(t-s)\mathcal{QL}}\mathcal{QL}u_{0}ds.
\end{align}
This equation may be interpreted as a mean field equation in 
the common situation where $\P$ is a conditional 
expectation. The two terms at the right 
hand side of \eqref{reduced order equation} are 
often called {\em streaming term} and 
{\em memory term}, respectively.

\subsection{Projection Operators}
\label{sec:proj}
The natural choice for the projection operator in 
the Mori-Zwanzig formulation is a 
conditional expectation \cite{Dominy2017}, i.e.,  a 
completely positive linear operator with suitable 
properties \cite{Umegaki}. Such conditional 
expectation can be rigorously defined in the context 
of operator algebras and it can have different forms.
Hereafter, we discuss the most important cases. 

\subsubsection{Chorin's Projection} In a series of papers 
\cite{Chorin,Chorin1,chorin2000optimal}, A. J. Chorin and 
collaborators defined the following projection operator
\begin{align}
\big(\mathcal{P}u\big)(\hat {\bm x}_{0})=\frac{\displaystyle \int_{-\infty}^{+\infty} u(\hat {\bm x}(t;\hat{\bm x}_0,\tilde{\bm x}_0),\tilde {\bm x}(t;\hat{\bm x}_0,\tilde{\bm x}_0))\rho_0(\hat {\bm x}_0,\tilde {\bm x}_0)d\tilde{\bm x}_0}{\displaystyle \int_{-\infty}^{+\infty}\rho_0(\hat {\bm x}_0,\tilde {\bm x}_0)d\tilde{\bm x}_0},
\label{Chorin_projection}
\end{align} 
 which represents a conditional expectation in the sense of classical 
 probability theory. In equation \eqref{Chorin_projection}, 
 $\bm x(t;x_0)$ denotes the flow map 
 generated by \eqref{eqn:nonautonODE},  
which we can split into 
resolved $\hat {\bm x}(t;\hat{\bm x}_0,\tilde{\bm x}_0)$ and 
unresoved $\tilde{\bm x}(t;\hat{\bm x}_0,\tilde{\bm x}_0)$ maps,  
$u(\bm x)=u(\hat {\bm x},\tilde {\bm x})$ is the quantity of interest, 
and $\rho_0(\hat{\bm x}_0,\tilde{\bm x}_0)$ is the probability density 
function  of the initial state $\bm x_0$. 
Alternatively, one can replace $\rho_0$ with 
 the equilibrium distribution of the 
 system $\rho_{eq}(\hat{\bm x},\tilde{\bm x})$, 
 assuming it exists.
Clearly, if $\bm x_0$ is deterministic then 
$\rho_0(\hat {\bm x}_0,\tilde {\bm x}_0)$ is a  product of 
Dirac delta functions. On the other hand, if $\hat{\bm x}_0$ and 
$\tilde{\bm x}_0$ are statistically independent, i.e. 
$\rho_{0}(\hat{\bm x}_0,\tilde{\bm x}_0) = \hat{\rho}_{0}(\hat{\bm x}_0)\tilde{\rho}_{0}(\tilde{\bm x}_0)$, then the conditional 
expectation \eqref{Chorin_projection} simplifies to
\begin{align}
\big(\mathcal{P}u\big)(\hat {\bm x}_0)=\int_{-\infty}^{+\infty}
u(\hat {\bm x}(t;\hat {\bm x}_0,\tilde{\bm x}_0),\tilde {\bm x}(t;\hat {\bm x}_0,\tilde{\bm x}_0))\tilde{\rho}_0(\tilde {\bm x}_0)d\tilde{\bm x}_0.
\label{8}
\end{align} 
In the special case where $u(\hat {\bm x},\tilde {\bm x})= 
\hat {\bm x}(t;\hat {\bm x}_0,\tilde{\bm x}_0)$ we have 
 \begin{align}
\big(\mathcal{P} \hat{\bm x}\big)(\hat {\bm x}_0)=
\int_{-\infty}^{+\infty}\hat {\bm x}(t;\hat {\bm x}_0,\tilde{\bm x}_0)
\tilde{\rho}_0(\tilde {\bm x}_0)d\tilde{\bm x}_0, 
\label{9}
\end{align} 
i.e. the conditional expectation of the resolved 
variables $\hat{\bm x}(t)$ given the initial condition 
$\hat{\bm x}_0$.  
This means that an integration of \eqref{9} 
with respect to $\hat{\rho}_0(\hat{\bm x}_0)$ 
yields the mean of the resolved variables 
\begin{equation}
\left< \hat{\bm x}(t)\right>_{\rho_0} = \int_{-\infty }^\infty \big(\mathcal{P} \hat{\bm x}\big)(\hat {\bm x}_0)\hat{\rho}_0(\hat{\bm x}_0)d\hat{\bm x}_0 = \int_{-\infty }^\infty \hat{\bm x}(t,\bm x_{0})\rho_{0}(\bm x_{0})d\bm x_{0}.
\end{equation}
Obviously, if the resolved variables $\hat{\bm x}(t)$ 
evolve from a deterministic initial state $\hat{\bm x}_0$ then 
the conditional expectation \eqref{9} represents 
the  average of the reduced-order flow 
map $\hat{\bm x}(t;\hat{\bm x}_0,\tilde{\bm x}_0)$ 
with respect to the 
PDF of $\tilde{\bm x}_0$, i.e., 
\begin{equation}
\bm X_0(t;\hat{\bm x}_0)=\int_{-\infty}^{+\infty}\hat {\bm x}(t;\hat {\bm x}_0,\tilde{\bm x}_0)\tilde{\rho}_0(\tilde {\bm x}_0)d\tilde{\bm x}_0.
\label{conditonal mean path}
\end{equation} 
In this case, the MZ equation \eqref{reduced order equation} is an unclosed 
evolution equation (PDE) for the averaged flow map \eqref{conditonal mean path}.

\subsubsection{Mori's Projection} 
Another definition of projection operator widely used 
in statistical mechanics is Mori's projection \cite{zwanzig2001nonequilibrium}
\begin{align}
\P u = \sum_{i=1}^M \frac{\langle u,\phi_i\rangle_{eq}}
{\langle\phi_i,\phi_i\rangle_{eq}}\phi_i(u(\bm x)).
\label{Mori_projection}
\end{align}
Here $\{\phi_1,\phi_2,...\}$ is an orthogonal basis 
that spans the Hilbert space of observables, i.e., functions of 
$u(\bm x)$ (assuming that such space is indeed a Hilbert space). 
Orthogonality  of $\{\phi_j\}$ is with respect to the inner product
\begin{align}
\langle a,b\rangle_{eq}=\int \rho_{eq}(\bm x) a(\bm x) b(\bm x) d \bm x,
\label{innerproduct}
\end{align}
where $a(\bm x)$, $b(\bm x)$ are two arbitrary 
phase space functions, while 
$\rho_{eq}$ is the equilibrium distribution function 
of the system, assuming it exists. 
In the context of Hamiltonian statistical mechanics the phase variables 
are $\bm x=(\bm p,\bm q)$, where $\bm q$ are generalized 
coordinates while $\bm p$ are kinetic momenta. 
In this setting the natural choice 
for $\rho_{eq}$ is the canonical Gibbs distribution 
\begin{equation}
\rho_{eq}(\bm p,\bm q)=\frac{1}{Z} e^{-\beta H(\bm p,\bm q)},
\label{gibbs}
\end{equation}
where $H(\bm p,\bm q)$ denotes the Hamiltonian of the system, 
and $Z$ is the partition function.


\subsubsection{Berne's Projection} 
A simpler projection operator  
was proposed by Berne in \cite{Berne} 
(see also \cite{Snook}, p. 30). The standard form is 
\begin{equation}
\P(\cdot) = \frac{\langle u_0, (\cdot)\rangle_{eq}}
{\langle u_0,u_0\rangle_{eq}}u_0.
\label{BernProjection}
\end{equation}
This projection can be considered as a subcase 
of the Mori projection \eqref{Mori_projection}.
Note that by using Berne's projection we can easily 
represent the auto-correlation function the observable $u(\bm x(t))$ as 
\begin{equation}
C_u(t)=\frac{\langle u(t),u_0\rangle_{eq}}{\langle u_0,u_0\rangle_{eq}} = \frac{\langle \P u(\bm x(t)),u_0 \rangle_{eq}}{\langle u_0,u_0\rangle_{eq}}.
\label{CorrelationFunction1}
\end{equation}

\section{Approximation of the Mori-Zwanzig Memory Integral}
\label{sec:New_expansion}
In this section, we develop new approximations of the Mori-Zwanzig 
memory integral 
\begin{equation}
\int_0^t \P e^{s\L}\P\L e^{(t-s)\Q\L}\Q\L u_0 ds
\label{MZmemory}
\end{equation}
based on series expansions of the orthogonal dynamics 
propagator $e^{t\Q\L}$ in the form 
\begin{equation}
e^{t\Q\L} = \sum_{n=0}^\infty a_n(t)\Phi_n\left(\Q\L\right),
\label{generalS}
\end{equation}
where $\Phi_n$ are polynomial basis functions,  and $a_n(t)$ 
are temporal modes. Series expansions in 
the form \eqref{generalS} can be rigorously defined 
in the context of matrix theory \cite{moler1978nineteen,moler2003nineteen}, i.e., 
for operators $\Q\L$ between finite-dimensional vector spaces. 
The question of whether it is possible to extend such 
expansions to the infinite-dimensional case, i.e., for operators 
acting between infinite-dimensional Hilbert or Banach spaces, 
is not a trivial \cite{dautray2012mathematical}. 
For example, it is known that the classical Taylor series 
\begin{equation}
e^{t\L}=\sum_{k=0}^\infty \frac{t^k}{k!}\L^k
\end{equation} 
does not hold if $\L$ is an unbounded operator, e.g., the 
generator of the Koopman semigroup \eqref{Koopman} (see \cite{Kato}, p. 481). 
In the latter case, $e^{t\L}$  should be properly defined as
\begin{equation}
e^{t\L}=\lim_{n\rightarrow \infty}\left(1-\frac{t\L}{n}\right)^{-n}.
\end{equation}
In fact, $\left(1-t\L /n\right)^{-1}$ is the resolvent of $\L$ (apart from a constant factor), which can be defined for both bounded and unbounded linear 
operators.
Despite the theoretical issues associated with 
the existence of convergent series expansions of 
semigroups generated by unbounded operators 
\cite{engel1999one,Kato}, when it comes to 
computing we always need to discretize the system, 
most often by discretizing the generator of the semigroup. 
In this setting, $e^{t\Q\L}$ is truly a matrix exponential, 
where, with some abuse of notation, we denoted by $\Q$ 
and  $\L$ the finite-dimensional representation\footnote{The 
matrix representation of a linear operator $\L$, relative 
to the span of a finite-dimensional basis  
$V=\textrm{span}\{h_1, h_2,...,\}$ 
can be easily obtained by representing 
each vector $\L h_i$ in $V$. Alternatively, if $\L$ 
operates in the Hilbert 
space $\mathcal{H}$ and $\{h_1, h_1,...,\}$ 
is an orthonormal basis of $\mathcal{H}$, then 
the matrix representation of $\L$ has entries 
$ \L_{ij}= (\L h_i,h_j)$, where $(,)$ denotes 
the inner product in $\mathcal{H}$.}
of the operators $\Q$ and $\L$.

\subsection{MZ-Dyson Expansion} 
\label{sec:MZ-Dyson} 
Consider the classical Taylor series expansion of the 
orthogonal dynamics propagator
\begin{align}
e^{t\Q\L}=\sum_{n=0}^{\infty}\frac{t^n}{n!}(\Q\L)^n.
\end{align}
A substitution of this expansion into the 
MZ equation \eqref{reduced order equation} yields 
\begin{align}
\frac{\partial}{\partial t}\P e^{t\L}u_0&=\P e^{t\L}\P\L u_0
+\int_0^t \P e^{s\L}\P\L e^{(t-s)\Q\L}\Q\L u_0 ds,\nonumber\\
&=\P e^{t\L}\P\L u_0+\int_0^t 
\sum_{n=0}^{\infty}\frac{(t-s)^n}{n!}\underbrace{\P e^{s\L}\P\L (\Q\L)^n\Q\L u_0}_{\mathcal{C}_n(s) u_0}ds,\nonumber\\
&=\P e^{t\L}\P\L u_0
+ \int_0^t \underbrace{\left[\sum_{n=0}^\infty  \mathcal{C}_n(s)\frac{(t-s)^n}{n!}\right]}_{\mathcal{G}(t-s,s)}
u_0ds,\nonumber\\ 
&=\P e^{t\L}\P\L u_0+\int_0^t \mathcal{G}(t-s,s) u_0ds,
\label{equ:GLE_dyson}
\end{align}
where the {\em memory operator}\footnote{Note that $\mathcal{G}(t-s,s)$ here is 
not a function but a linear operator.} $\mathcal{G}(t-s,s)$ is defined as
\begin{equation}
\mathcal{G}(t-s,s)=\sum_{n=0}^{\infty}\frac{(t-s)^n}{n!}\mathcal{C}_n(s), 
\qquad \mathcal{C}_n(s)=\P e^{s\L}\P\L(\Q\L)^n\Q\L, \quad  n\geq 0.
\label{MZDg}
\end{equation}
We shall call this series expansion of the MZ equation 
as {\em MZ-Dyson expansion}. The reason for such definition is that \eqref{equ:GLE_dyson} is equivalent to the $H$-model 
discussed in \cite{Yuan1} and \cite{stinis2007higher}, which in turn is equivalent to a Dyson series expansion in the form
\begin{align*}
\frac{\partial}{\partial t}\P e^{t\L}u_0&=\P e^{t\L}\P\L u_0
+w_0(t)
\end{align*}
where 
\begin{align}
\label{wn(t)}
w_0(t)=&\int_0^t\mathcal{P}
e^{s\mathcal{L}}\mathcal{PLQL}x_0ds
+
\int_0^t\int_0^{\tau_1}\mathcal{P}
e^{s\mathcal{L}}\mathcal{PLQLQL}x_0dsd\tau_1\nonumber\\
&+...+
\int_0^t\int_0^{\tau_{n-1}}...\int_0^{\tau_1}\mathcal{P}
e^{s\mathcal{L}}\mathcal{PL}(\mathcal{QL})^nx_0ds
d\tau_1...d\tau_{n-1}+...\, .
\end{align}
To prove such equivalence, we just need to prove that 
\begin{align}\label{equ_Dyson}
\int_0^t\int_0^{\tau_{n-1}}...\int_0^{\tau}\P e^{s\L}\P\L(\Q\L)^ndsd\tau_1...d\tau_{n-1}=
\int_0^t\frac{(t-s)^{n-1}}{(n-1)!}\P e^{s\L}\P\L(\Q\L)^n\Q\L ds.
\end{align}
We proceed by induction. To this end, we first define 
\begin{equation}
\mathcal{A}_n(t)=\int_0^t\int_0^{\tau_{n-1}}...\int_0^{\tau}\P e^{s\L}\P\L(\Q\L)^ndsd\tau_1...d\tau_{n-1}, 
\qquad 
\mathcal{B}_n(t)=\int_0^t\frac{(t-s)^{n-1}}{(n-1)!}\P e^{s\L}\P\L(\Q\L)^n\Q\L ds.
\label{AnBn}
\end{equation}
For $n=1$ we have $\mathcal{A}_1=\mathcal{B}_1$.
For  $n\geq 2$ we have $\mathcal{A}_n'(t)=\mathcal{A}_{n-1}(t)$,  
$\mathcal{B}_n'(t)=\mathcal{B}_{n-1}(t)$ and 
$\mathcal{A}_n(0)=\mathcal{B}_n(0)$. 
Hence, by induction we conclude that 
$\mathcal{A}_n(t)=\mathcal{B}_n(t)$, and therefore the 
memory integral in \eqref{equ:GLE_dyson}, with 
$\mathcal{G}$ given in \eqref{MZDg}, is equivalent 
to a Dyson series.

\subsection{MZ-Faber Expansion} 
\label{sec:MZ-Faber}
The Faber series of the orthogonal 
dynamics propagator $e^{t\Q\L }$ is an operator 
series in the form (see Appendix \ref{app:Faber}) 
\begin{equation}
e^{t\Q\L} = \sum_{j=0}^\infty a_j(t)\mathcal{F}_j(\Q\L),
\label{expQL}
\end{equation}
where $\mathcal{F}_j$ is the $j-$th order Faber 
polynomial, and $a_j(t)$ are suitable temporal modes defined 
hereafter. 
The series expansion \eqref{expQL} 
is {\em asymptotically optimal}, 
in the sense that its $m$-th order truncation 
uniformly approximates the best sequence 
of operator polynomials converging to $e^{t\Q\L}$ 
as $m\rightarrow \infty$ \cite{Eiermann}.  A substitution of \eqref{expQL}
into \eqref{MZmemory} yields the following expansion 
of the MZ equation \eqref{reduced order equation}
\begin{align}\label{volterra_MZequation}
\frac{\partial}{\partial t}\P e^{t\L}u_0&=\P e^{t\L}\P\L u_0
+\int_0^t \P e^{s\L}\P\L e^{(t-s)\Q\L}\Q\L u_0 ds,\nonumber\\
&=\P e^{t\L}\P\L u_0
+ \int_0^t \sum_{j=0}^\infty  a_j(t-s)
 \underbrace{\P e^{s\L}\P\L \mathcal{F}_j(\Q\L)\Q\L u_0}_{\mathcal{C}_j(s)u_0}
 ds,\nonumber\\
&=\P e^{t\L}\P\L u_0
+ \int_0^t \mathcal{G}(t-s,s)u_0ds,
\end{align}
where
\begin{align}
\mathcal{G}(t-s,s)=\sum_{j=0}^{\infty}a_j(t-s) \mathcal{C}_j(s),
\label{memoryK}
\end{align}
and 
\begin{equation}
a_j(t-s)=\frac{1}{2\pi i}\int_{|w|=R}\frac{e^{(t-s)\psi(w)}}{w^{j+1}} dw, \qquad 
\mathcal{C}_j(s)= \P e^{s\L}\P\L \F_j(\Q\L)\Q\L.
\label{coeff}
\end{equation}
Here, $\psi(w)$ is the conformal map at the basis 
of the Faber series (see Appendix \ref{app:Faber}). 
The coefficients of the Laurent expansion of $\psi$ determine 
the recurrence relation of the Faber polynomials. 
High-order Laurent series usually yield higher convergence 
rates, but complicated recurrence relations (see equation \eqref{recursive}). 
Moreover, the computation of the integrals in \eqref{coeff} can 
be quite cumbersome if high-order Laurent series are employed.
To avoid such  drawbacks, in this paper we choose 
the conformal map $\psi(w)=w+c_0+c_1/w$. 
This yields the following expression for the 
coefficients $a_j(t-s)$ 
\begin{equation}
a_j(t-s)= \frac{e^{(t-s)c_0}}{(\sqrt{-c_1})^j}
J_j\left(2(t-s)\sqrt{-c_1}\right),
\label{aaj}
\end{equation}
where $J_j$ denotes the $j-$th Bessel function of the first kind. 
In Section \ref{sec:convergence} we prove that the Faber 
expansion of the MZ memory integral converges for 
any linear dynamical system and any finite integration 
time with rate that is at least $R$-superlinear.

\paragraph{Remark} The MZ-Dyson expansion we 
discussed in Section \ref{sec:MZ-Dyson} is a subcase 
of the Faber expansion. 
In fact, Faber polynomials $\F_j(\Q\L)$ corresponding 
to the conformal mapping $\psi(w)=w$ are 
simply monomials $(\Q\L)^j$ (see Appendix \ref{app:Faber}).
Moreover, the temporal modes \eqref{aaj} reduce to $(t-s)^j/j!$ 
if we set $c_0=0$ and take the limit $c_1\rightarrow 0$.

\subsection{Other Series Expansions of the MZ-Memory Integral}
The operator exponential $e^{t\Q\L}$ (propagator of the orthogonal 
dynamics) can be expanded relative to basis functions other 
than Faber polynomials \cite{moler1978nineteen,moler2003nineteen}. 
This yields different approximations of the MZ memory integral 
and, correspondingly, different expansions of the MZ equation. Hereafter 
we discuss two relevant cases. 

\subsubsection{MZ-Lagrange Expansion} 
The MZ-lagrange expansion is based on the following 
semigroup expansion
\begin{align}
e^{t\Q\L}=\sum_{j=1}^{n}e^{\lambda_j t}\prod_{\substack{k=1\\k\neq j}}^{n}\frac{(\Q\L-\lambda_k\mathcal{I})}{(\lambda_j-\lambda_k)},
\label{MZL}
\end{align}
where $\{\lambda_1,...,\lambda_{n}\}=\sigma(\Q\L)$ 
is the spectrum of the matrix representation of the 
operator $\Q\L$ (eigenvalues counted with their multiplicity). 
Note that \eqref{MZL} is in the form
\eqref{generalS} with 
\begin{equation}
a_j(t)=e^{\lambda_j t}, \quad \textrm{and}\quad 
\Phi_j(\Q\L)=\prod_{\substack{k=1\\k\neq j}}^{n}\frac{(\Q\L-\lambda_k\mathcal{I})}{(\lambda_j-\lambda_k)}.
\end{equation}
A substitution of \eqref{MZL} into the MZ equation yields the 
MZ-Lagrange expansion
\begin{align}
\frac{\partial}{\partial t}\P e^{t\L}u_0=\P e^{t\L}\P\L u_0+\int_0^t 
\mathcal{G}(t-s,s) u_0ds,
\label{equ:GLE_lagrange}
\end{align}
where 
\begin{equation}
\mathcal{G}(t-s,s)=\sum_{j=1}^{n}e^{(t-s)\lambda_j}\mathcal{C}_j(s), \quad 
\textrm{and}
\quad 
\mathcal{C}_j(s)= \P e^{s\L}\P\L \prod_{\substack{k=1\\k\neq j}}^{n}\frac{(\Q\L-\lambda_k\mathcal{I})}{(\lambda_j-\lambda_k)},\qquad j\geq 1.
\label{Lagrangekernel}
\end{equation}

\subsubsection{MZ-Newton Expansion}
The MZ-Newton expansion is based on the following 
semigroup expansion
\begin{align}
e^{t\Q\L}=f_{1,1}(t)\mathcal{I}+\sum_{j=2}^n
f_{1,j}(t)\prod_{k=1}^{j-1}(\Q\L-\lambda_k\mathcal{I}),
\label{MZN}
\end{align}
where $f_{1,j}(t)$ is the divided difference defined recursively by
\begin{align}
f_{1,j}(t)=
\begin{dcases}
e^{\lambda_1 t}\quad &j=1,\\
\displaystyle \frac{e^{t\lambda_1}-e^{t\lambda_2}}{\lambda_1-\lambda_2} \quad & j=2,\\
\displaystyle \frac{f_{1,j-1}(t)- f_{2,j}(t)}{\lambda_1-\lambda_{j}} \quad &  j\geq 3.
\end{dcases}
\label{mznewton}
\end{align}
A substitution of the Newton expansion \eqref{MZN} into the MZ equation 
yields the following MZ-Newton expansion
\begin{align}
\frac{\partial}{\partial t}\P e^{t\L}u_0=\P e^{t\L}\P\L u_0+
\int_0^t \mathcal{G}(t-s,s) u_0ds,
\label{equ:GLE_newton}
\end{align}
where 
\begin{equation}
\mathcal{G}(t-s,s)=\mathcal{C}_1(s)e^{(t-s)\lambda_1}+\sum_{j=2}^n\mathcal{C}_j(s) f_{1,j}(t),\quad  
\quad \mathcal{C}_j(s)= 
\begin{cases}
\P e^{s\L}\P\L&\quad j=1\\
\displaystyle \P e^{s\L}\prod_{k=1}^{j-1}(\P\L\Q\L-\lambda_k\mathcal{PL})& \quad 
j\geq 2
\end{cases}.
\end{equation}

\begin{table}[t]
\begin{minipage}{4.5cm}
\small
\vspace{0.3cm}Mori-Zwanzig Memory Operator \\
$\displaystyle \mathcal{G}(t-s,s)=\sum_{j=0}^{\infty} h_j(t-s) \mathcal{C}_j(s)$\\
\end{minipage}
\hspace{0.5cm}
\begin{minipage}{5cm}
\label{table}
\centering\small
\begin{tabular}{l l l}
Type & Temporal bases $h_j(t)$ & Operators $\mathcal{C}_j(s)$\vspace{0.1cm} \\
\hline\\
MZ-Dyson     & $\displaystyle \frac{t^j}{j!}$   & $\P e^{s\L}\P\L(\Q\L)^j\Q\L$ 
\\
MZ-Faber         & $\displaystyle e^{tc_0}\frac{J_j(2t\sqrt{-c_1})}{(\sqrt{-c_1})^j}$   &  $\P e^{s\L}\P\L \mathcal{F}_j(\Q\L)\Q\L$\\
MZ-Lagrange & $\displaystyle e^{t \lambda_j}$ & 
$\displaystyle \P e^{s\L}\P\L \prod_{\substack{k=1\\k\neq j}}^{n}\frac{(\Q\L-\lambda_k\mathcal{I})}{(\lambda_j-\lambda_k)}$\\
MZ-Newton & $f_{1,j}(t)$ & 
$\displaystyle 
\begin{dcases}
\P e^{s\L}\P\L &\quad j=1\\
\displaystyle \P e^{s\L}\prod_{k=1}^{j-1}(\P\L\Q\L-\lambda_k\mathcal{PL})& \quad 
j\geq 2
\end{dcases}
$\\
\end{tabular}
\end{minipage}
\caption{Series expansions of the Mori-Zwanzig memory operator. Here $J_j$ is the $j$th Bessel function of the first kind, $c_0$ and $c_1$ are real numbers, 
$f_{1,j}(t)$ are defined in \eqref{mznewton}, and $\lambda_j$ are the eigenvalues of any matrix representation of $\Q\L$.}
\label{tab:1}
\end{table}

\noindent
\paragraph{Remark} All series expansion methods we considered 
so far aim at representing the memory integral in the Mori-Zwazing equation 
for the same phase space function. Therefore, such series 
should be related to each other. Indeed, as shown in 
Table \ref{tab:1}, they basically represent the same memory 
operator $\mathcal{G}(t-s,s)$ relative to different bases. This also means 
that the series can have different convergence rate. For example, 
as we will demonstrate numerically in Section \ref{sec:application} 
the MZ-Faber expansion converges much faster than the MZ-Dyson series.

\subsection{Generalized Langevin Equation}
\label{sec:form_GLE}
We have seen in Section \ref{sec:New_expansion} 
that expanding the orthogonal dynamics 
propagator $e^{t\Q\L}$ in an operator series 
in the form \eqref{generalS}
yields the Mori-Zwanzig equation 
\begin{align}\label{general_GLE}
\frac{\partial}{\partial t}\P e^{t\L}u_0=\P e^{t\L}\P\L u_0+ \sum_{j=0}^{\infty}\int_0^t h_j(t-s)\mathcal{C}_j(s)u_0ds,
\end{align}
where $h_j(t-s)$ are temporal modes, and $\mathcal{C}_j(s)$ 
are operators defined in Table \ref{tab:1}. For example, if 
we consider the MZ-Dyson expansion, we have
\begin{equation}
h_j(t-s)=\frac{(t-s)^j}{j!},\qquad 
\mathcal{C}_j(s)=\P e^{s\L}\P\L(\Q\L)^j\Q\L.
\end{equation}
Equation \eqref{general_GLE} is the exact
generalized Langevin equation (GLE) governing 
the projected dynamics of a quantity of interest. 
Such equation has different forms depending on the 
choice of the projection operator $\P$. In particular, 
if we choose Chorin's projection \eqref{Chorin_projection}
then \eqref{general_GLE} is an equation for the conditional 
expectation of the quantity of interest. On the other 
hand, if we choose Berne's projection \eqref{BernProjection}
then \eqref{general_GLE} becomes an equation for the 
autocorrelation function of the quantity of interest. 

\subsubsection{Evolution Equation for the Conditional Expectation}
\label{sec:MZmean}

If we consider Chorin's projection \eqref{Chorin_projection}, 
then \eqref{general_GLE} becomes an unclosed 
evolution equation for the conditional expectation 
of the quantity of interest (see Section \ref{sec:proj}).
However, in the special case where the dynamical 
system \eqref{eqn:nonautonODE}
is {\em linear} and the quantity of interest is $u(\bm x)=x_1(t)$, 
it can be shown that the evolution equation for the conditional 
expectation is closed. 
To this end,  let us first recall that if $\P$ is Chorin's projection 
and $u(\bm x)=x_1$ then 
\begin{align*}
\P e^{t\L}x_1(0)=\langle x_1(t)\rangle_{\rho_0} = \int x_1(t,\bm x_0) \rho_0(\bm x_0)d\bm x_0.
\end{align*}
In this case,  \eqref{general_GLE} reduces to 
\begin{align}\label{equ:mean_evolution1}
\frac{d}{d t}\langle x_1(t)\rangle_{\rho_0}&=
a\langle x_1(t)\rangle_{\rho_0}+b + \int_0^t g(t-s)\langle x_1(s)\rangle_{\rho_0}ds+\int_0^t f(t-s)ds,
\end{align}
where the constants $a$, $b$, the {\em MZ memory kernel} $g(t-s)$, 
and the function $f(t-s)$ are defined by
\begin{equation}
\P\L x_1(0)=ax_1(0)+b,\qquad 
g(t-s) =\sum_{j=0}^{\infty}g_jh_j(t-s),\qquad f(t-s) =\sum_{j=0}^{\infty}f_jh_j(t-s).
\label{param1}
\end{equation}
The coefficients $g_j$, $ f_j$ and the temporal bases $h_j(t-s)$ 
appearing in the series expansions above depend on the series 
expansion of the orthogonal dynamics propagator $e^{t\Q\L}$. 
Specifically, $g_j$ and $f_j$ are determined by the equation 
\begin{equation}
\mathcal{C}_j(s)x_1(0)=g_j\langle x_1(s)\rangle_{\rho_0}+f_j,
\label{eq44}
\end{equation}
while $h_j(t-s)$ and $\C_j(s)$ are defined in Table \ref{tab:1}.
To derive equation \eqref{eq44} we used the identity $\P e^{s\L}f_j=f_j$. 
In the case of MZ-Dyson and MZ-Faber expansions  
we explicitly obtain
\begin{align}
\P\L(\Q\L)^j\Q\L x_1(0)=g_j^Dx_1(0)+f_j^D,\qquad \qquad
\P\L\F_j(\Q\L)\Q\L x_1(0)=g_j^Fx_1(0)+f_j^F,
\label{GLEcoeff0}
\end{align}
where the superscripts $D$ and $F$ stand for ``Dyson'' and ``Faber'', 
respectively.

\subsubsection{Evolution Equation for the Autocorrelation Function }
\label{sec:autocorrelation}
If we choose the projection operator $\P$ to be Berne's 
projection \eqref{BernProjection}, then equation \eqref{general_GLE} 
becomes a closed evolution equation for the autocorrelation 
function $C_u(t)$ of the quantity of interest. Such equation has the form 
\begin{equation}
\frac{d C_u(t) }{d t}=  a C_u(t)+\int_0^tg(t-s)C_u(s)ds,
\label{equ:evolution_correlation}
\end{equation}
where $a$ and $g$ are defined as
\begin{equation}
\P \L u_0=au_0,\qquad 
g(t-s) = \sum_{j=0}^{\infty}g_jh_j(t-s).
\label{param}
\end{equation}
As before, the temporal modes $h_j$ and the coefficients 
$g_j$ in the expansion of the MZ-memory kernel $g(t-s)$ depend 
on the expansion of the orthogonal dynamics propagator 
$e^{t\Q\L}$. Specifically, in the case of MZ-Dyson and MZ-Faber expansions 
we obtain, respectively,
\begin{align}
\P\L(\Q\L)^j\Q\L u_0=g_j^Du_0,\qquad 
\P\L\F_j(\Q\L)\Q\L u_0=g_j^Fu_0.
\label{GLEcoeff}
\end{align} 
It is worth noticing that Berne's projection sends any function into the linear space spanned by the initial condition $u_0$. 

\subsubsection{Analytical Solution to the Generalized Langevin Equation}
The analytical solution to the MZ equations 
\eqref{equ:mean_evolution1} and \eqref{equ:evolution_correlation} 
can be computed through Laplace transforms. 
To this end, let us first notice that both equations are 
in the form of a Volterra equation
\begin{equation}
\frac{d y(t) }{d t}=  a y(t)+b+\int_0^tg(t-s)y(s)ds+\int_0^tf(t-s)ds.
\label{Volterra}
\end{equation}
Applying the Laplace transform
\begin{equation}
\mathscr{L}[\cdot](s)=\int_0^\infty  (\cdot)  e^{-st}dt 
\end{equation} 
 to  both sides of \eqref{Volterra} yields 
\begin{equation}
s Y(s)-y(0)= a Y(s)+\frac{b}{s}+ Y(s)G(s)+\frac{F(s)}{s},
\label{LaplaceT}
\end{equation}
i.e., 
\begin{equation}
Y(s)=\frac{(F(s)+b)/s+y(0)}{s-G(s)-a},
\label{Ys}
\end{equation}
where 
\begin{equation}
Y(s)=\mathscr{L}[y(t)], \qquad F(s)=\mathscr{L}[f(t)],\qquad 
G(s)=\mathscr{L}[g(t)].
\end{equation}
Thus, the exact solution to the Volterra 
equation \eqref{Volterra} can be written as
 \begin{align}
y(t)=\mathscr{L}^{-1}\left[\frac{(F(s)+b)/s+y(0)}{s-G(s)-a}\right].
\label{MZ_solution1}
\end{align}
The Laplace transform of the memory kernel $g(t)$, i.e., $G(s)$, 
can be computed analytically in many cases. 
For example, in the case of MZ-Dyson and MZ-Faber 
expansions we obtain, respectively 
\begin{align}\label{equ:lapalce_MZDyson}
G(s)=\sum_{j=0}^{\infty}\frac{g^D_j}{s^{j+1}}\qquad 
\textrm{(MZ-Dyson),}
\end{align}
\begin{align}\label{equ:lapalce_MZFaber}
G(s)=\sum_{j=0}^{\infty}\frac{g^F_j}{2^j(\sqrt{-c_1})^{2j}}
\frac{(\sqrt{s^2-4c_1}-s)^j}{\sqrt{s^2-4c_1}} \qquad \textrm{(MZ-Faber).}
\end{align}
The coefficients $g^F_j$ and $g^D_j$ are explicitly 
defined in \eqref{GLEcoeff0}, or \eqref{GLEcoeff}, 
depending on whether we are interested in the mean 
or the correlation function of the quantity of interest.

\paragraph{Remark} The recurrence 
relation at the basis of the Faber polynomials 
(see equation \eqref{recursive}) induces a recurrence relation in 
the Laplace transform $G(s)$ of the MZ memory kernel.  
Therefore, a connection between the MZ-Faber 
approximation method we propose here and the 
method of recurrence relations of Lee 
\cite{mori1965continued,lee1982eq} can 
be established.

\section{Convergence Analysis}
\label{sec:convergence}
In this section, we develop a thorough  convergence 
analysis of the MZ-Faber expansion\footnote{We recall that 
the MZ-Dyson series expansion is 
a subcase of the MZ-Faber expansion. 
Therefore convergence of MZ-Faber implies 
convergence of MZ-Dyson.} of the Mori-Zwanzig 
equation \eqref{reduced order equation}.   
The key theoretical 
results at the basis of our analysis can be found 
in our recent paper \cite{Yuan1}. Here we focus, in particular, 
on high-dimensional linear systems in the form 
\begin{equation}
\dot {\bm x}(t) = \bm A \bm x(t),\qquad \bm x(0)=\bm x_0(\omega),
\label{lin}
\end{equation}
where $\bm x_0(\omega)$ is a random initial state. 
Our goal is to prove that the norm of the 
approximation error    
\begin{align}
E_n(t)&=\int_0^t \P e^{s\L}\P\L e^{(t-s)\Q\L}\Q\L u_0 ds
-\underbrace{\sum_{j=0}^n\int_0^t  a_j(t-s)\P e^{s\L}
 \P\L \mathcal{F}_j(\Q\L)\Q\L u_0ds}_{\textrm{MZ-Faber series}}\nonumber \\
 &=\int_0^t \sum_{j=n+1}^\infty  a_j(t-s)\P e^{s\L}
 \P\L \mathcal{F}_j(\Q\L)\Q\L u_0 ds
 \label{Ent}
\end{align}
decays as we increase the polynomial order $n$, 
for any fixed integration time $t>0$, i.e., 
\begin{align*}
\lim_{n\rightarrow \infty} \left\|E_n(t)\right\|=0.
\end{align*} 
Throughout this Section $\left\|\cdot\right\|$ denotes either 
an operator norm, a norm in a function space or 
a standard norm in $\mathbb{C}^N$, depending 
on the context.  
The convergence proof of MZ-Faber series 
clearly depends on the choice of the projection 
operator and the phase space function $u(\bm x)$ (quantity 
of interest). In this Section, we consider  
\begin{equation}
u(\bm x(t))=x_1(t),
\label{PSF}
\end{equation}
and Chorin's projection \eqref{Chorin_projection}. Similar results 
can be obtained for Berne's projection. 
We begin with the following 

\begin{lemma}
\label{lemma_matrix_op}
Consider the linear dynamical system \eqref{lin} and the phase space function \eqref{PSF}. Then,
\begin{align*}
\P\L p_k(\Q\L)\Q\L x_{1}(0)&=\left[\bm b \cdot p_k(\Mi_{11}^T) \bm a \right] 
x_{1}(0)+ \left[p_k(\Mi_{11}^T)\Mi_{11}^T\bm a\right]\cdot \langle \bm {x}_{-1}(0)\rangle_{\rho_0},
\end{align*}
where $\P$ is Chorin's projection \eqref{Chorin_projection}, 
$\Q=\I-\P$, $\L=\bm A\bm x\cdot \nabla$, $p_k$ is an arbitrary polynomial of degree $k$,  
\begin{align*}
\bm {x}_{-1}(0)=[x_2(0), x_3(0), \ldots, x_N(0)]^T\qquad 
\bm a = [ A_{12}, \dots, A_{1N}]^T,\qquad 
\bm b&= [A_{21},\dots, A_{N1}]^T,
\end{align*}
and $\Mi_{11}$ is the matrix obtained from $\bm A$ by removing the first 
row and the first column. 
\end{lemma} 

\begin{proof}
By a direct calculation, it can be verified that 
\begin{align}
(\Q\L)^{n}x_{1}(0) & = \left[\left({\bm M}_{11}^{T}\right)^{n-1}\bm{a} \right]
\cdot \left[ \bm x_{-1}(0)-\langle \bm x_{-1}(0)\rangle_{\rho_0}\right],  
\nonumber\\
\L(\Q\L)^{n}x_1(0) & = \left[\bm{b}^{T}\left({\bm M}_{11}^{T}\right)^{n-1}
\bm{a}\right]x_1(0)+
\left[\left({\bm M}_{11}^{T}\right)^{n}\bm{a}\right]\cdot \bm x_{-1}(0),
\nonumber\\
\P\L(\Q\L)^{n}\Q\L x_{1}(0)&=\left[\bm b^T\left(\Mi_{11}^T\right)^{n} \bm a\right] x_{1}(0)+
\left[ \left(\Mi_{11}^T\right)^{n}\Mi_{11}^T\bm a\right]\cdot \langle\bm {x}_{-1}(0)\rangle_{\rho_0}.
\label{formula:operator_chorin}
\end{align}
Note that each entry of the vector $\langle \bm x_{-1}(0)\rangle_{\rho_0}=[\langle x_2(0)\rangle_{\rho_0},...,\langle x_N(0)\rangle_{\rho_0}]^T$ is 
$\langle x_i(0)\rangle_{\rho_0}=\P x_{i}(0)$ ($i=2,...,N$).
Thus, for any polynomial function in the form 
\begin{equation}
p_k(\Q\L) = \sum_{j=0}^k \beta_k (\Q\L)^j,
\end{equation}
 we have
\begin{align*}
\P\L p_k(\Q\L)\Q\L x_{1}(0)=&\sum_{j=0}^k\beta_j\P\L(\Q\L)^j\Q\L x_{1}(0),
\nonumber\\
=&\sum_{j=0}^k\beta_j
\left(
\left[\bm b^T\left(\Mi_{11}^T\right)^{n} \bm a\right] x_{1}(0)+
\left[ \left(\Mi_{11}^T\right)^{n}\Mi_{11}^T\bm a\right]\cdot 
\langle\bm {x}_{-1}(0)\rangle_{\rho_0}\right),\nonumber\\
=&\left[\bm b \cdot p_k\left(\Mi_{11}^T\right) \bm a \right] 
x_{1}(0)+ \left[p_k\left(\Mi_{11}^T\right)\Mi_{11}^T\bm a\right]\cdot 
\langle \bm {x}_{-1}(0)\rangle_{\rho_0}.
\end{align*}
This completes the proof of the Lemma. 
\end{proof}
To prove convergence of MZ-Faber series we need 
two more Lemmas involving Faber polynomials 
in the complex plane (see Appendix \ref{app:Faber}). 
\begin{lemma}\label{lemma_psi}
Let $\gamma$ be the capacity of $\Omega\subseteq \mathbb{C}$. 
If $\Omega$ is symmetric with respect to the real axis,  
then for any $R>\gamma$ the conformal map \eqref{mapping}
satisfies
\begin{align*}
\psi(R)\leq \psi(\gamma)+R-\frac{\gamma^2}{R}.
\end{align*} 
\end{lemma}
\begin{proof}
We first notice that 
\begin{align*}
\psi(R)=\psi(\gamma)+\int_{\gamma}^R \psi'(t)dt.
\end{align*}
By using Lemma 4.2 in \cite{novati2003solving}, i.e., 
\begin{align*}
|\psi'(t)|\leq 1+\left(\frac{\gamma}{|t|}\right)^2,\quad |t|>\gamma
\end{align*}
we have 
\begin{align*}
\psi(R)-\psi(\gamma)\leq |\psi(R)-\psi(\gamma)|= \left|\int_{\gamma}^R\psi'(t)dt\right|\leq \int_{\gamma}^R|\psi'(t)|dt=R-\frac{\gamma^2}{R},
\end{align*}
which completes the proof.
\end{proof}

\noindent
Next, consider an arbitrary matrix $\bm A$ and 
define the {\em field value} of $\bm A$ as 
\begin{align*}
FV(\Ai)=\left\{\bm z^H\Ai \bm z : \bm z\in \mathbb{C}^N,\,   \bm z^H\bm z=1\right\}.
\end{align*}
The field value of $\bm A$ is a subset of the complex plane. 
Also, denote the truncated Faber series of the exponential matrix $e^{t\bm A}$ as 
\begin{align}\label{method7}
\bm P_{m}(t)=\sum_{j=0}^{m} a_j(t) \F_j(\Ai).
\end{align}
With this notation, we have the following
\begin{lemma}\label{lemma_error}
Let $\Omega\subset \mathbb{C}$ 
be symmetric with respect to the real axis, convex and with capacity $\gamma$. 
Consider an $N\times N$ matrix $\Ai$ with spectrum $\sigma(\bm A)$, 
and an $N\times 1$ vector $\bm v$. If $\sigma(\bm A)\subseteq \Omega$ and 
the field value $FV(\Ai)\subseteq\Omega(q)$ for some $q\geq \gamma$, then the approximation error 
\begin{equation*}
\bm e_{m}(t-s)\bm v=  e^{(t-s)\Ai}\bm v -\bm P_{(m-1)}(t-s) \bm v
\qquad t\geq s
\end{equation*}
satisfies  
\begin{align*}
\|\bm e_{m}(t-s)\bm v\|\leq C_3e^{(t-s)E}\left(\frac{qe^{t-s}}{m} \right)^{m-1}\quad m\geq 4q,
\end{align*}
where 
\begin{align*}
C_3=C_3(v)=8e\|\bm v\|q\left(1+\frac{1}{8q}\right)\quad \textrm{and}
\quad E=1+\psi(\gamma).
\end{align*}
\end{lemma}
\begin{proof}
If $q\geq \gamma$ then we have, 
thanks to the convexity of $\Omega$ and the analyticity 
of the exponential function, 
\begin{align}\label{foruma_1}
\|\bm e_{m}(t-s)\bm v\|\leq 8\|\bm v\|e\left(1+\frac{1}{8q}\right)m\left(\frac{q}{m}\right)^m\max_{|z|\in\Gamma(m)}\left| e^{(t-s)z}\right| \qquad m\geq 4q
\end{align}
(see Theorem 4.2 in \cite{novati2003solving}). 
On the other hand,
\begin{align}\label{foruma_2}
\max_{|z|\in\Gamma(m)}\left| e^{(t-s)z}\right| =e^{(t-s)\psi(m)}\qquad m\geq 4q.
\end{align}
By using Lemma \ref{lemma_psi} we have 
\begin{equation}
\psi(m)\leq \psi(\gamma)+m-\frac{\gamma^2}{m}\leq \psi(\gamma)+m,
\end{equation}
and therefore  
\begin{align}\label{foruma_3}
e^{(t-s)\psi(m)}\leq e^{(t-s)(m-1)}e^{(t-s)(1+\psi(\gamma))}\quad m\geq 4q\geq\gamma.
\end{align}
Combining \eqref{foruma_1} , \eqref{foruma_2} and \eqref{foruma_3}, 
we obtain
\begin{align}\label{estimation}
\|\bm e_{m}(t-s)\bm v\|\leq C_3\exp((t-s)E)\left(\frac{qe^{t-s}}{m} \right)^{m-1},
\end{align}
where 
\begin{align*}
C_3=8e\|\bm v\|q\left(1+\frac{1}{8q}\right)\quad \textrm{and}\quad 
E=1+\psi(\gamma).
\end{align*}
\end{proof}
\noindent 
At this point, we we have all elements to prove the following 
\begin{theorem}
\label{Theorem}
{\bf (Convergence of the MZ-Faber Expansion)}
Consider the linear dynamical system \eqref{lin},
the phase space function \eqref{PSF} and the projection 
operator \eqref{Chorin_projection}. The norm of 
the approximation error \eqref{Ent} satisfies\footnote{It can be shown 
that the upper bound in \eqref{ub} is always positive.}
\begin{align}
\|E_n(t)\|\leq K\left(\frac{q}{n+1}\right)^{n}
\frac{e^{t\beta}-e^{t(E+n)}}{\beta-E-n}\qquad t\geq 0,\quad n\geq 4q,
\label{ub}
\end{align}
where $n$ is the Faber polynomial order, while $q$,  $K$, $\beta$ and $E$ 
are suitable constants defined in the proof of the theorem.
\end{theorem}

\begin{proof}
We aim at finding an upper bound for  
\begin{align}
\|E_n(t)\|&=\left\|\int_{0}^t\P e^{s\L}\sum_{j=n+1}^\infty  a_j(t-s)
 \P\L \F_j(\Q\L)\Q\L x_{1}(0)ds\right\|.
 \label{pr1}
\end{align}
To this end, we fist notice that quantity 
$\F_j(\Q\L)\Q\L$ is a $(j+1)$-th order operator 
polynomial in $\Q\L$. Thus, we can apply 
Lemma \ref{lemma_matrix_op} to obtain
\begin{align}\label{formula_DP}
\P\L \F_j(\Q\L)\Q\L x_{1}(0)=
\left[\bm b \cdot \F_j(\Mi_{11}^T) \bm a \right] 
x_{1}(0)+ \left[\F_j(\Mi_{11}^T)\Mi_{11}^T\bm a\right]\cdot \langle \bm {x}_{-1}(0)\rangle_{\rho_0}.
\end{align}
Let us now set
\begin{align}
\eta_n(t-s)=\left\|\sum_{j=n+1}^{\infty}a_j(t-s)
\P\L \F_j(\Q\L)\Q\L x_{1}(0)\right\|.
\label{eta}
\end{align}
By using \eqref{formula_DP} and the Cauchy-Schwartz inequality 
we have  
\begin{align}
\eta_n(t-s)\leq 
C_4\left\|\sum_{j=n+1}^{\infty}a_j(t-s)\F_j({\bm M}_{11}^T) \bm a\right\|+C_5
\left\|\sum_{j=n+1}^{\infty}a_j(t-s)\F_j(\bm M_{11}^T) \Mi_{11}^T\bm a\right\|,
\label{sums}
\end{align}
where $C_4=\|\bm b^T\||x_1(0)|$, $C_5=\|\langle\bm {x}_{-1}(0)\rangle_{\rho_0}\|$. 
The two sums in \eqref{sums} represent the error in the Faber approximation 
of the matrix exponential $e^{(t-s){\bm M}^T_{11}}$. In fact, 
\begin{align}
\bm e_{(n+1)}(t-s)=e^{(t-s){\bm M}^T_{11}} -\sum_{j=1}^{n}a_j(t-s)
\F_j({\bm M}^T_{11})
=\sum_{j=n+1}^{\infty}a_j(t-s)\F_j(\bm M^T_{11}).
\end{align}
Combining \eqref{pr1}, \eqref{eta}, \eqref{sums} and \eqref{estimation} yields  
\begin{align}
\|E_n(t)\|&\leq\int_{0}^t \eta_n(t-s)\left\|\P e^{s\L}\right\|ds,\nonumber\\
&\leq \int_0^t
\left(C_4\left\|\bm e_{(n+1)}(t-s)\bm a\right\|+C_5\left\|\bm e_{(n+1)}(t-s)\bm 
M_{11}^T\bm a \right\|\right)\|\P e^{s\L}\|ds,\nonumber\\
&\leq\int_0^t Ke^{s\beta}e^{(t-s)(E+n)}\left(\frac{q}{n+1}\right)^{n}ds,
\nonumber\\
&\leq K\left(\frac{q}{n+1}\right)^{n}\frac{e^{t\beta}-e^{t(E+n)}}{\beta-E-n}
\quad n\geq 4q.
\label{fin} 
\end{align}
Here we used the semigroup estimation $\|e^{s\L}\|\leq W  e^{s\beta}$. 
The constants in \eqref{fin} are   
\begin{equation}
K=\|\P\|C_6W, \qquad C_6=2\max\{C_4C_3,C_5C_3^*\}, \qquad E=1+\psi(\gamma),
\end{equation}
where 
\begin{equation}
C_3=8e\|\bm a\|q\left(1+\frac{1}{8q}\right),\qquad 
C_3^*=8e\|\bm M_{11}^T \bm  a\|q\left(1+\frac{1}{8q}\right).
\end{equation}
It can be shown that the upper bound \eqref{fin} 
is always positive, and goes to zero as we send the 
Faber polynomial order $n$ to infinity. This implies that 
\begin{equation}
\lim_{n\rightarrow \infty} \left\|E_n(t)\right\|=0,
\end{equation}
i.e., the MZ-Faber expansion converges for 
any finite time $t\geq 0$. This completes the proof. 
\end{proof}
Next, we estimate the convergence rate of the MZ-Faber expansion. 
To this end, let us define 
\begin{equation}
R(t,n)=K\left(\frac{q}{n+1}\right)^{n}\frac{e^{t\beta}-e^{t(E+n)}}{\beta-E-n},\quad n\geq 4q
\label{R}
\end{equation}
the be the upper bound \eqref{fin}. We have the following 
\begin{corollary}
{\bf (Convergence Rate of the MZ-Faber Expansion) }
With the same the notation of Theorem \ref{Theorem}, 
the MZ-Faber expansion converges at least $R$-superlinearly 
with the polynomial order, i.e.
\begin{align}
\lim_{n\rightarrow \infty}\frac{R(t,n+1)}{R(t,n)}=0
\label{Rconv}
\end{align}
for any finite time $t\geq 0$.
\end{corollary}
\begin{proof}
By a direct calculation it is easy to verify that \eqref{Rconv} holds true. In fact,
\begin{align}
\frac{R(t,n+1)}{R(t,n)}=\frac{q}{n+2}\left(\frac{n+1}{n+2}\right)^n\frac{e^{t\beta}-e^{t(E+n+1)}}{e^{t\beta}-e^{t(n+E)}}\frac{\beta-E-n}{\beta-E-(n+1)}
\end{align}
Therefore\footnote{We recall that
\begin{equation}
\lim_{n\rightarrow+\infty}\left(\frac{n+1}{n+2}\right)^n=\frac{1}{e}.
\end{equation}
 },
\begin{align}
\lim_{n\rightarrow+\infty}\frac{R(t,n+1)}{R(t,n)}=
\lim_{n\rightarrow+\infty}\frac{qe^{t}}{n+2}\left(\frac{n+1}{n+2}\right)^n=0, \qquad t<\infty.
\label{rateFab}
\end{align}
\end{proof}
By using asymptotic analysis we can show theoretically that
also the MZ-Dyson expansion converge $R$-superlinearly. To this end, 
let us define the MZ-Dyson approximation error  
\begin{equation}
E_n(t)=\int_0^t \P e^{s\L}\P\L e^{(t-s)\Q\L}\Q\L u_0 ds
-\underbrace{\sum_{j=0}^n\int_0^t  a_j(t-s)\P e^{s\L}
 \P\L (\Q\L)^j\Q\L u_0ds}_{\textrm{MZ-Dyson series}}.
 \label{mzdys}
\end{equation}
By following the same steps we used in the proof of Theorem \ref{Theorem}, 
we can bound the norm of \eqref{mzdys}  as 
\begin{equation}
\left\|E_n(t)\right\|\leq F(t,n).
\end{equation}
where 
\begin{equation}
F(t,n) = C\frac{(At)^n}{(n+1)!}\qquad A, C \geq 0.
\end{equation}
Such upper bound plays the same role as $R(t,n)$ in the MZ-Faber expansion of 
$E_n(t)$ (see Eqs. \eqref{fin} and \eqref{R}). 
Taking the ratio between $F(t,n+1)$ and $F(t,n)$ we obtain 
\begin{equation}
\lim_{n\rightarrow \infty}\frac{F(t,n+1)}{F(t,n)} =\lim_{n\rightarrow \infty}\frac{At}{n+2}=0.
\label{rateDys} 
\end{equation}

\section{Numerical Examples}
\label{sec:application}
In this section, we demonstrate the accuracy and 
effectiveness of the MZ-Dyson and MZ-Faber expansion 
methods we developed in this paper in applications 
to prototype problems involving random wave propagation 
and harmonic chains of oscillators interacting on a Bethe lattice.

\subsection{Random Wave Propagation}
\label{sec:waves}
Consider the following initial/boundary value problem for 
the wave equation in an annulus with radii $r_1=1$ and $r_2=11$ 
\begin{align}
\label{wave_annulus}
\frac{\partial^2w}{\partial t^2}=\frac{\partial ^2w}{\partial r^2}+\frac{1}{r}\frac{\partial w}{\partial r}+\frac{1}{r^2}\frac{\partial ^2w}{\partial \theta^2}, 
\end{align}
where 
\begin{align}
w(t,r_1,\theta)=0, \qquad w(t,r_2,\theta)=0\qquad w(0,r,\theta)=w_0(r,\theta;\omega),
\qquad  \frac{\partial w(0,r,\theta)}{\partial t}=0.
\end{align}
The field $w(t,r,\theta)$ represents the wave amplitude 
at time $t$, while $w_0(r,\theta;\omega)$ is the wave field at initial time, 
which is set to be random. We seek the for an approximation of the solution 
$w(t,r,\theta)$ in the form 
\begin{align}
\label{w_form}
w_N(t,r,\theta)=\sum_{n=1}^N\widehat w_n(t)\psi_n(r,\theta),
\end{align}
where $\psi_n(r,\theta)$ are standard trigonometric functions. 
The random wave field at initial time is represented as 
\begin{equation}
w_0(r,\theta;\omega)=\sum_{n=1}^M\widehat w_n(0)\psi_n(r,\theta),\qquad M\leq N,
\label{IC}
\end{equation}
where $\widehat w_n(0)$ are i.i.d Gaussian random variables.
We substitute \eqref{w_form} into \eqref{wave_annulus} and impose 
that the residual is orthogonal to the space spanned by 
the basis $\{\psi_1,...,\psi_N\}$  \cite{Hesthaven}. 
This yields the linear system 
\begin{align}\label{equ:mode_space1}
\frac{d^2}{dt^2}\widehat {\bm w}(t)=\bm A\widehat{ \bm w}(t),
\end{align}
where $\bm A$ is an $N\times N$ matrix with entries
\begin{equation}
A_{mn}=\frac{\displaystyle \int_{r_1}^{r_2}\int_{0}^{2\pi}\left(\frac{\partial ^2\psi_n}{\partial r^2}+\frac{1}{r}\frac{\partial \psi_n}{\partial r}+\frac{1}{r^2}\frac{\partial ^2\psi_n}{\partial \theta^2}\right)\psi_m drd\theta}{
\displaystyle \int_{r_1}^{r_2}\int_{0}^{2\pi}\psi_m^2 drd\theta}.
\end{equation}
We are interested in building  a {\em convergent} 
reduced-order model for the wave amplitude at a 
specific point within the annulus, e.g., where we 
placed a sensor. To this end, we transform the 
system \eqref{equ:mode_space1} form the modal 
space to the nodal space defined by an interpolant of 
at $N$ collocation points. 
Such transformation can be easily defined by 
evaluating \eqref{w_form} at a set of distinct 
collocation nodes $\bm x_n=(r_{i(n)},\theta_{j(n)})$ 
($n=1,...,N$) within the annulus. This yields 
\begin{align}\label{tranformation}
\bm w(t)=\bm \Psi\widehat {\bm w}(t),
\end{align}
where $\bm w(t)=[w(t,\bm x_1),...,w(t,\bm x_N)]^T$, while $\bm \Psi$ is the 
$N\times N$ transformation matrix defined as 
\begin{align*}
\bm \Psi=\left[
\begin{matrix}
\psi_1(\bm x_1)&\dots&\psi_N(\bm x_1)\\
\vdots&&\vdots\\
\psi_1(\bm x_N)&\dots&\psi_N(\bm x_N)
\end{matrix}
\right].
\end{align*}
Differentiating \eqref{tranformation} with respect to time 
we obtain 
\begin{align}\label{equ:nodel}
\frac{d^2}{dt^2}\bm w(t)=\bm \Psi \Ai\bm \Psi^{-1}\bm w(t).
\end{align}
This system evolves from the random initial state 
\begin{equation}
\bm w(0)=\bm \Psi \widehat {\bm{w}}(0), \qquad \frac{d \bm w(0)}{dt}=0.
\end{equation}
In Figure \ref{fig:wave_annulus} we plot the mean solution of the 
random wave equation for initial conditions in the form \eqref{IC} 
with different number of modes. 
\begin{figure}
\centerline{
\includegraphics[height=4cm]{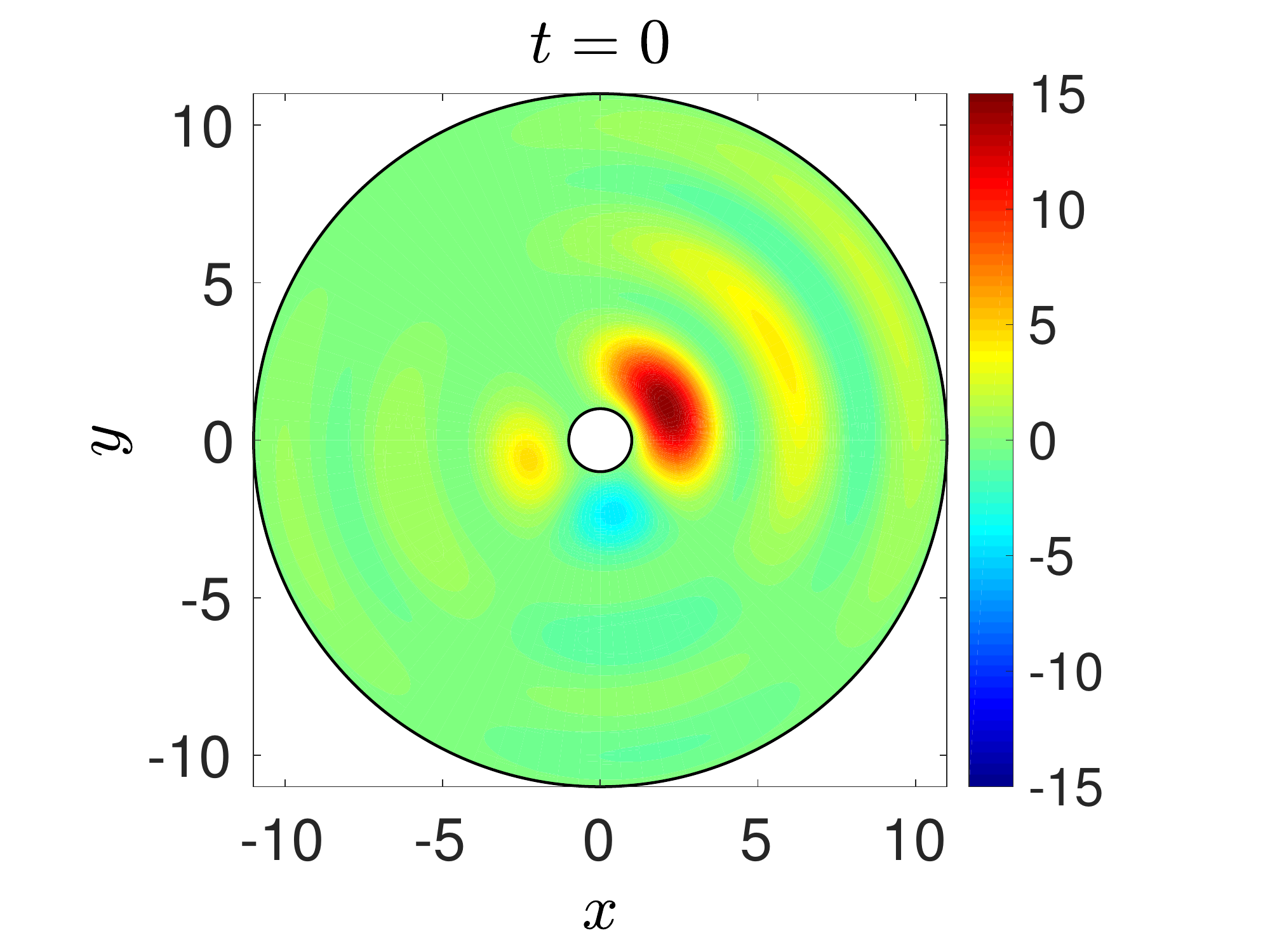} 
\includegraphics[height=4cm]{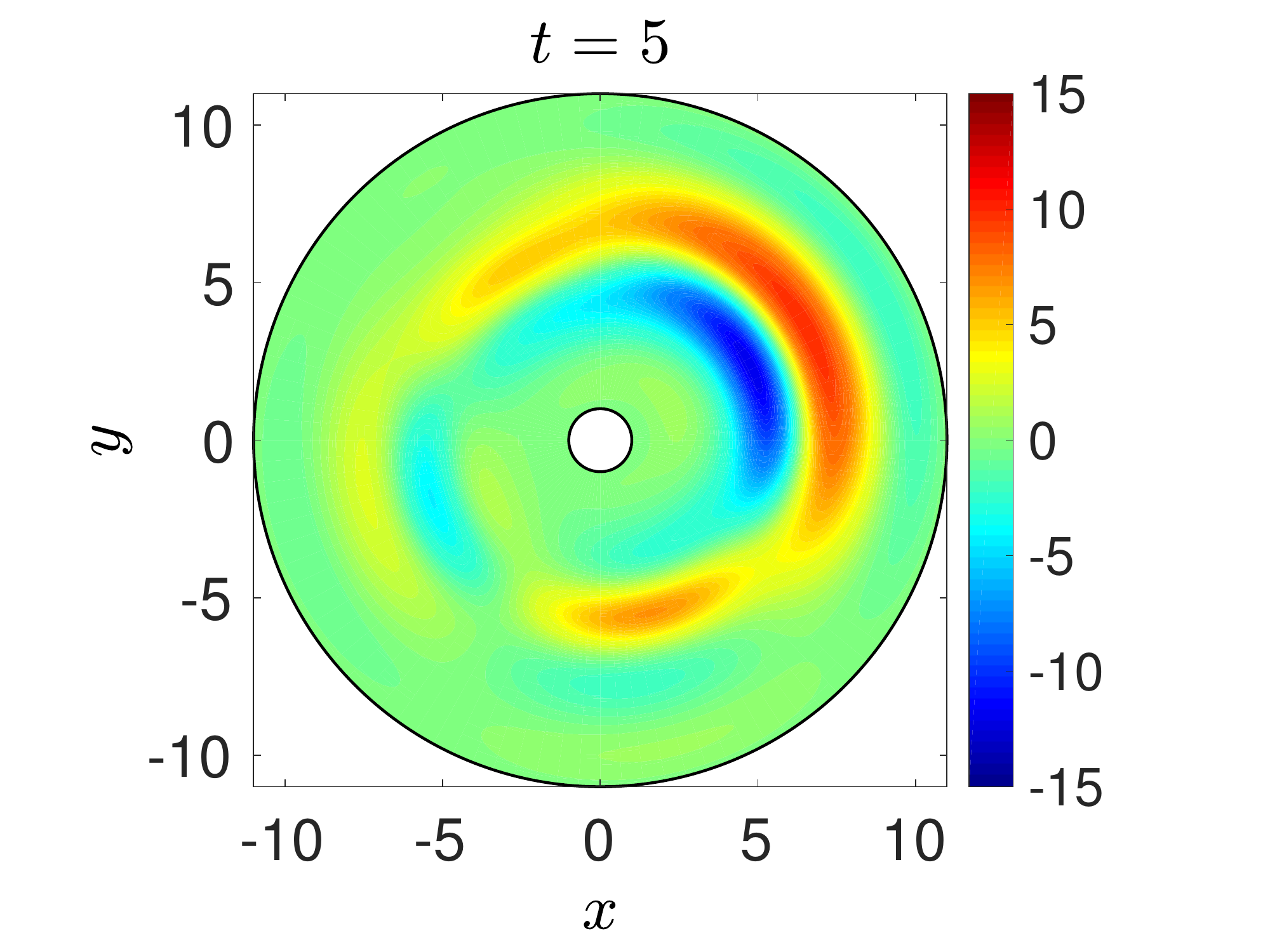}
\includegraphics[height=4cm]{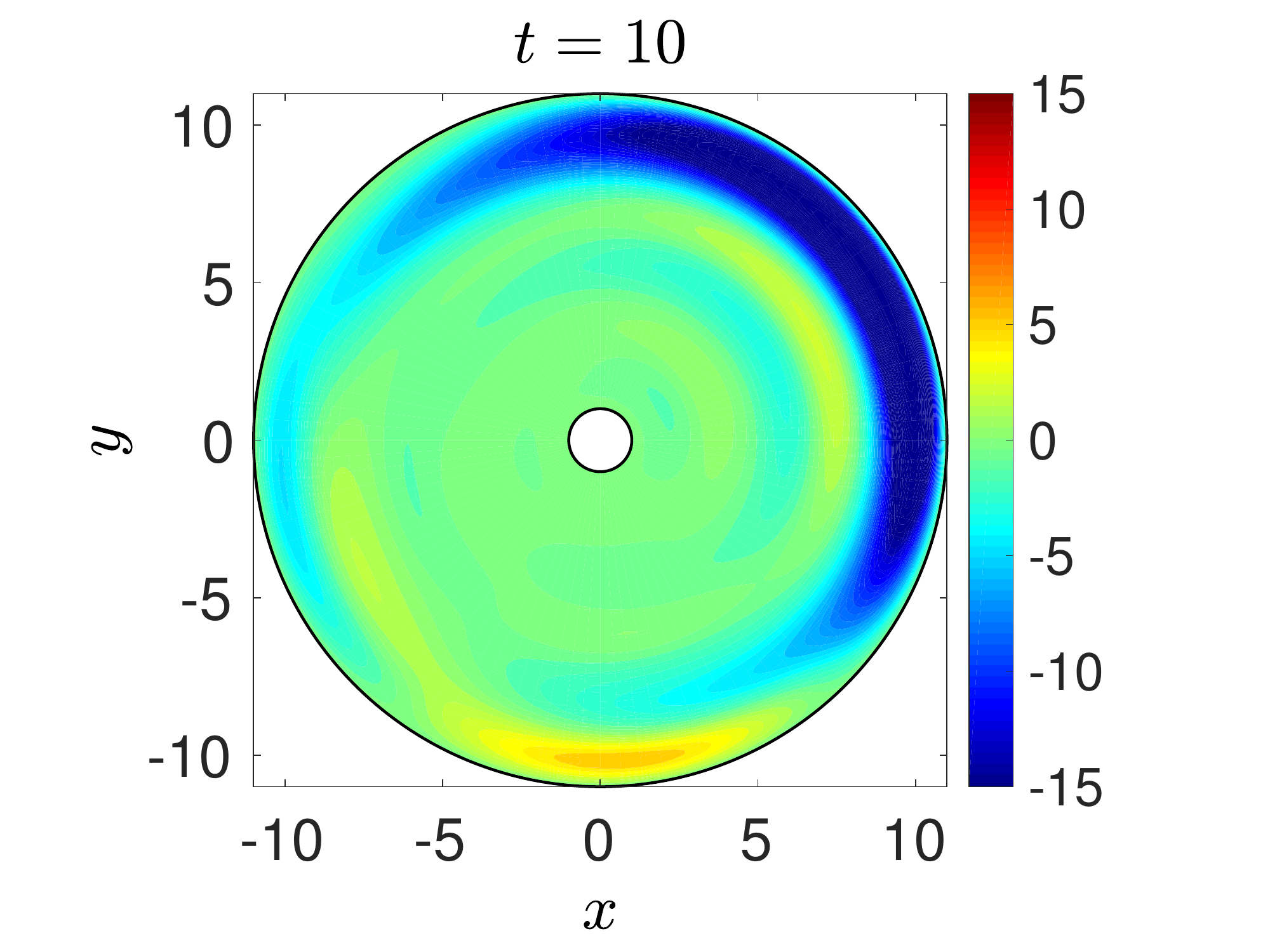}}
\centerline{
\includegraphics[height=4cm]{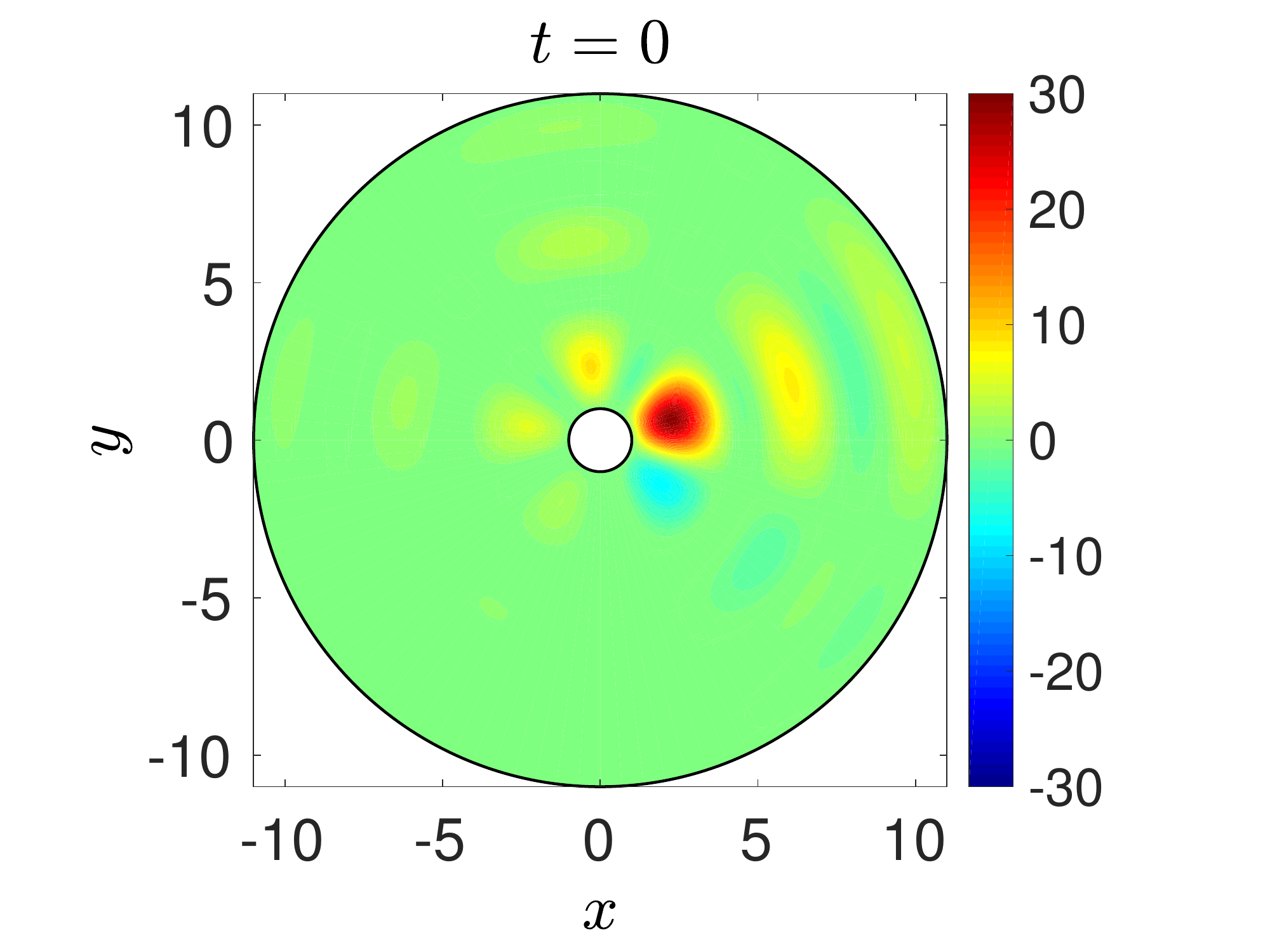} 
\includegraphics[height=4cm]{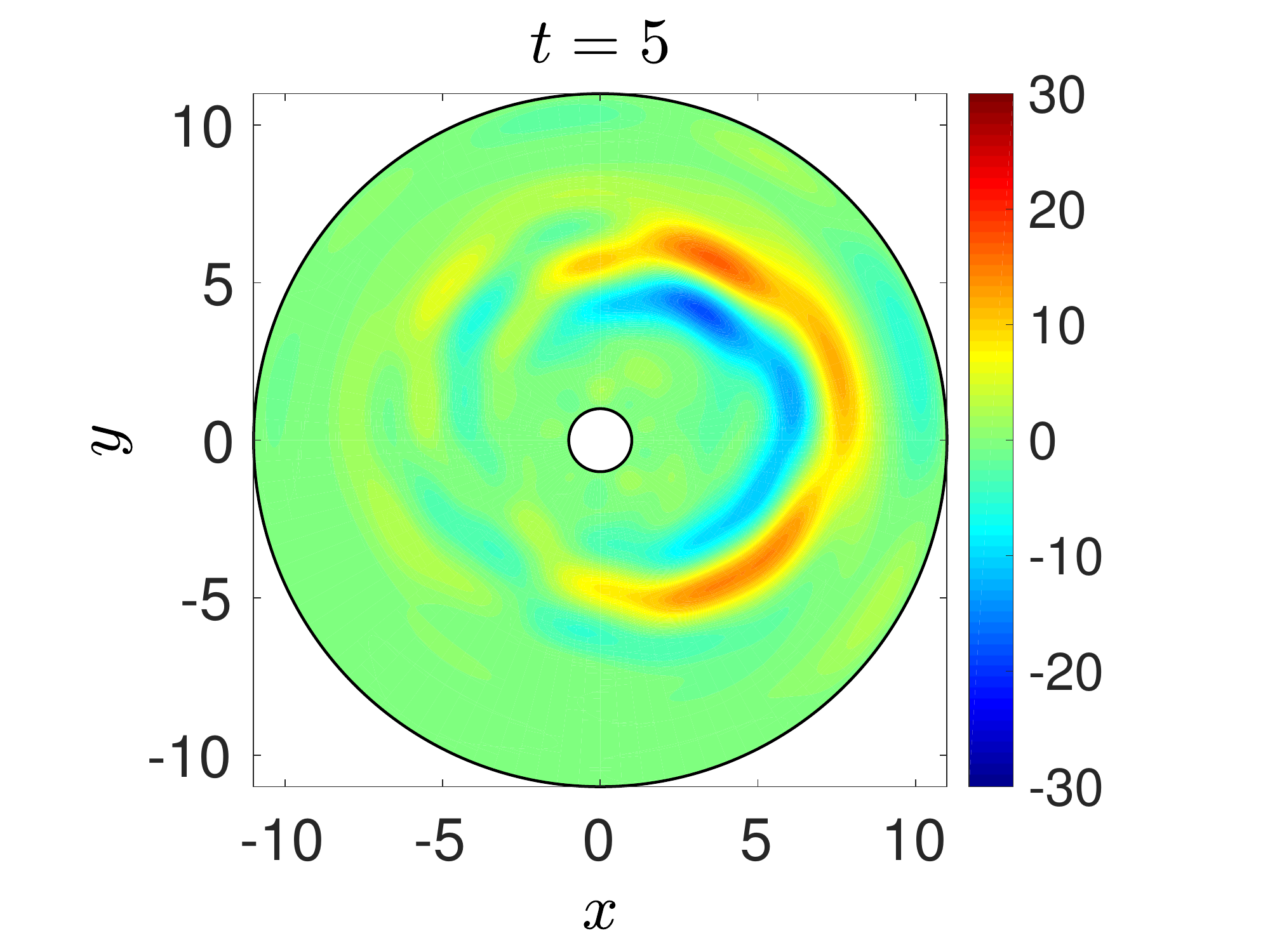}
\includegraphics[height=4cm]{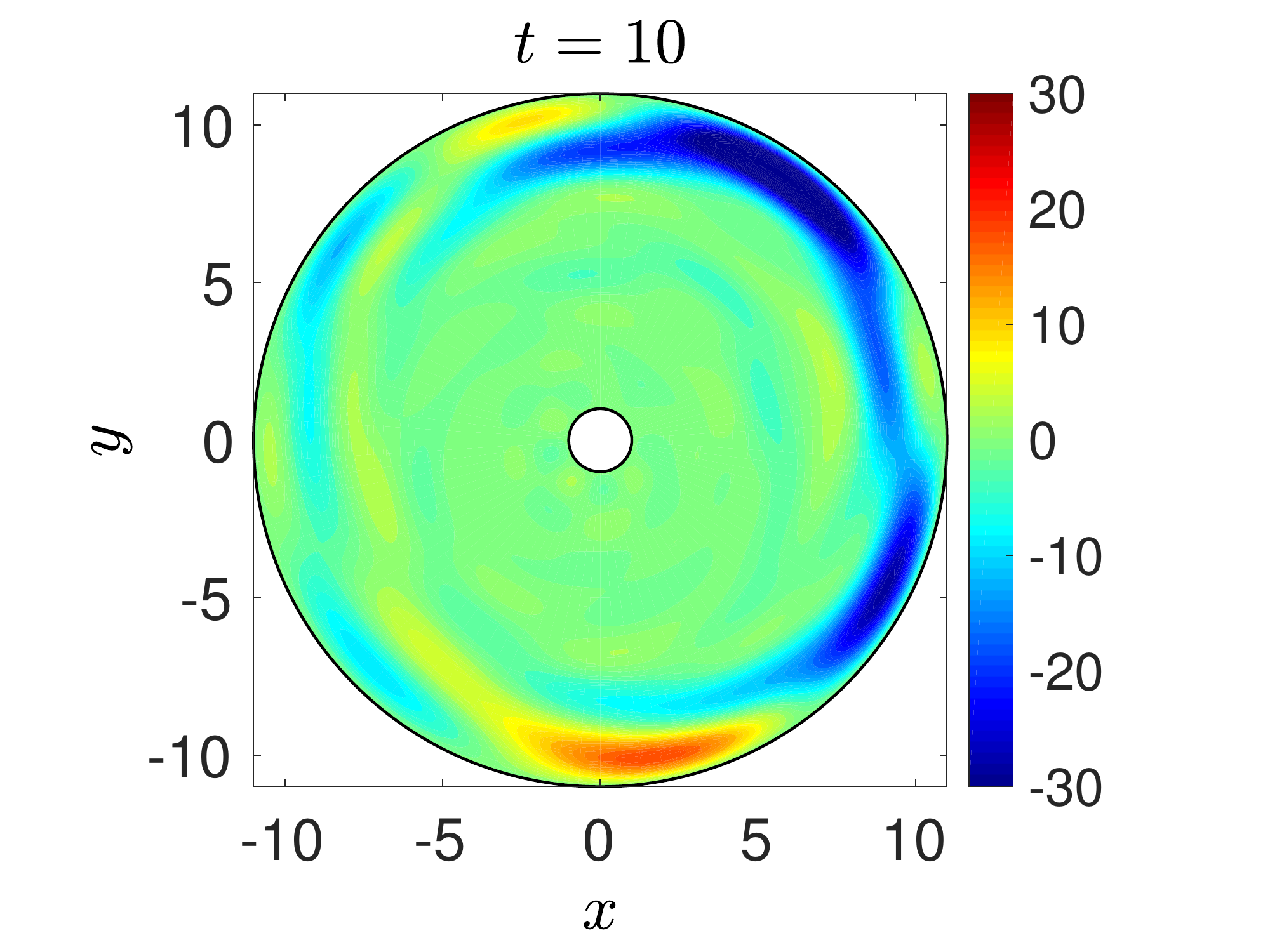}} 
\caption{Mean solution of the random wave equation in the annulus. We consider 
two random initial conditions in the form \eqref{IC}, with different number of 
modes: $M=25$ (first row), $M=50$ (second row).}
\label{fig:wave_annulus}
\end{figure}

\paragraph{Generalized Langevin Equation for the Mean Wave Amplitude}

We are interested in building  a convergent 
reduced-order model for the mean wave amplitude 
at a specific point within the annulus, e.g., where we 
would like to place a sensor. Such a dynamical system can 
be constructed  by using the Mori-Zwanzig formulation 
and Chorin's projection operator \eqref{Chorin_projection}. 
In particular, let us define the quantity of interest as
$u(\bm w) = w_1(t)$, i.e., the wave amplitude 
at the spatial point $(r,\theta)=(1.1,0.1)$. 
The exact evolution equation for mean of $w_1(t)$ was 
derived in Section \ref{sec:MZmean}, and it is 
rewritten hereafter for convenience 
\begin{align}\label{equ:mean_evolution}
\frac{d}{d t}\langle w_1(t)\rangle_{\rho_0}&=a\langle w_1(t)\rangle_{\rho_0}+b+
\int_0^t g(t-s)\langle w_1(s)\rangle_{\rho_0}ds+\int_0^t f(t-s)ds.
\end{align}
We recall that 
\begin{align*}
\P\L w_1(0) = & B_{11}w_1(0)+\bm a\cdot\bm \left<\bm w_{-1}(0)\right>_{\rho_0} \\
= & aw_1(0)+b,
\end{align*} 
and $\L=\left[\bm \Psi \bm A\bm \Psi^{-1} \bm w\right] \cdot \nabla$. 
The memory kernel $g(t-s)$ and the function $f(t-s)$ 
can be expanded by using in any of the operator series 
summarized in Table \ref{tab:1}. For instance, if we employ 
MZ-Faber series we obtain
\begin{align}
g(t-s) =\sum_{j=0}^n g^F_j e^{tc_0} \frac{J_j(2t\sqrt{-c_1})}{(\sqrt{-c_1})^j},
\qquad
f(t-s) = \sum_{j=0}^n f^F_j e^{tc_0} \frac{J_j(2t\sqrt{-c_1})}{(\sqrt{-c_1})^j}.
\label{MZFk}
\end{align}
The coefficients $g_j^F$ and $f_j^F$ are explicitly obtained as 
\begin{align}
g_j^F=\bm b^{T}\F_j\left(\bm{M}_{11}^{T}\right)\bm a, \qquad 
f_j^F = \left[ \F_j\left(\bm {M}_{11}^{T}\right)\Mi_{11}^T\bm a\right] \cdot \langle\bm {w}_{-1}(0)\rangle_{\rho_0},
\end{align}
where
\begin{align*}
\bm {w}_{-1}(0)=[w_2(0), w_3(0), \ldots, w_N(0)]^T,\qquad 
\bm a = [ B_{12}, \dots, B_{1N}]^T,\qquad 
\bm b&= [B_{21},\dots, B_{N1}]^T,
\end{align*}
 $\bm B=  \bm \Psi\bm A\bm \Psi^{-1}$ and $\Mi_{11}$ is 
 the matrix obtained from $\bm B$ by removing the 
 first row and the first column. In Figure \ref{fig:memory}
 we study convergence of MZ-Dyson and
 MZ-Faber series expansions of the memory kernel.
\begin{figure}[t]
\centerline{\hspace{0.5cm}MZ-Dyson\hspace{6.5cm}MZ-Faber}
\centerline{
\includegraphics[height=6cm]{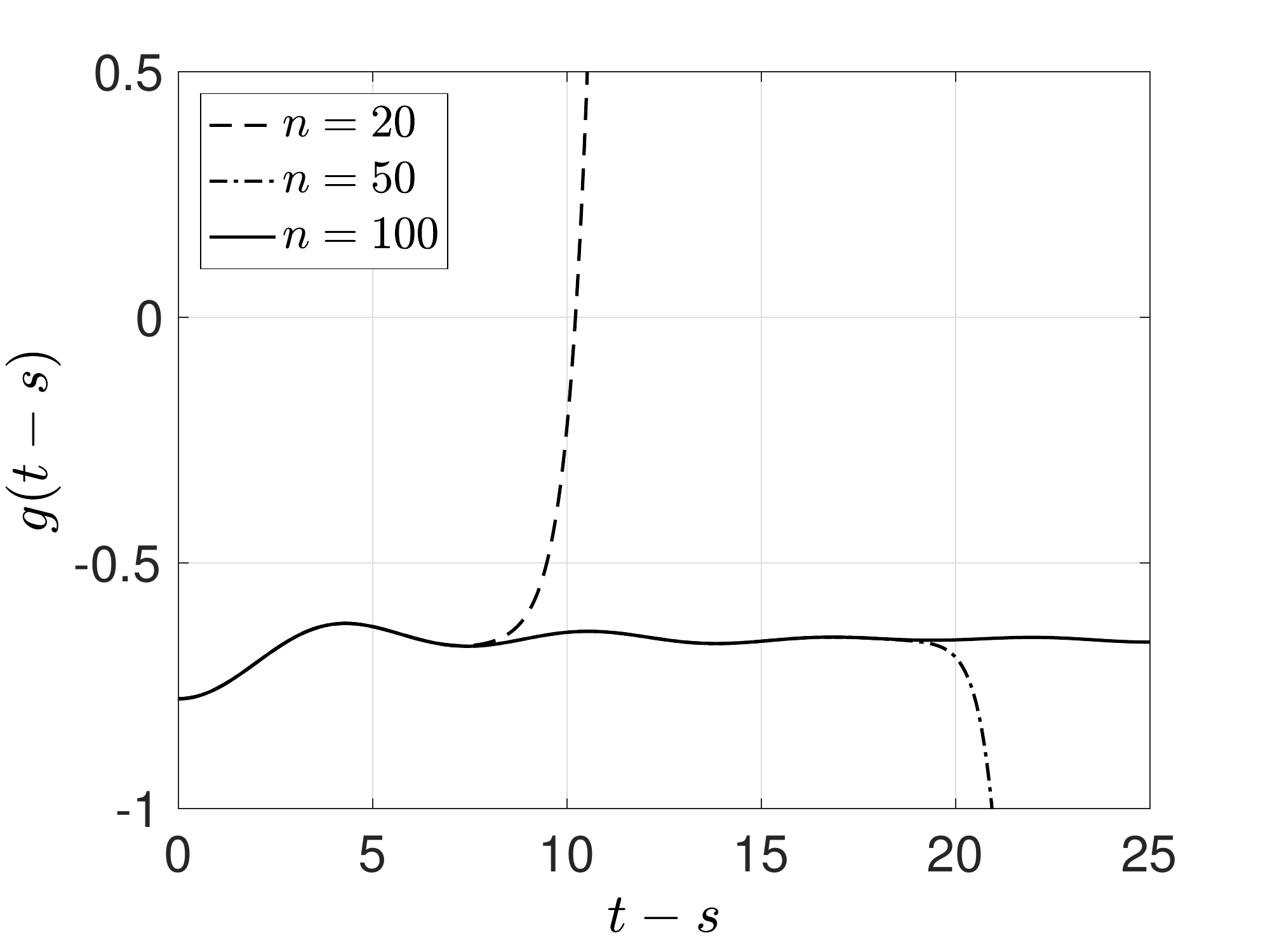} 
\includegraphics[height=6cm]{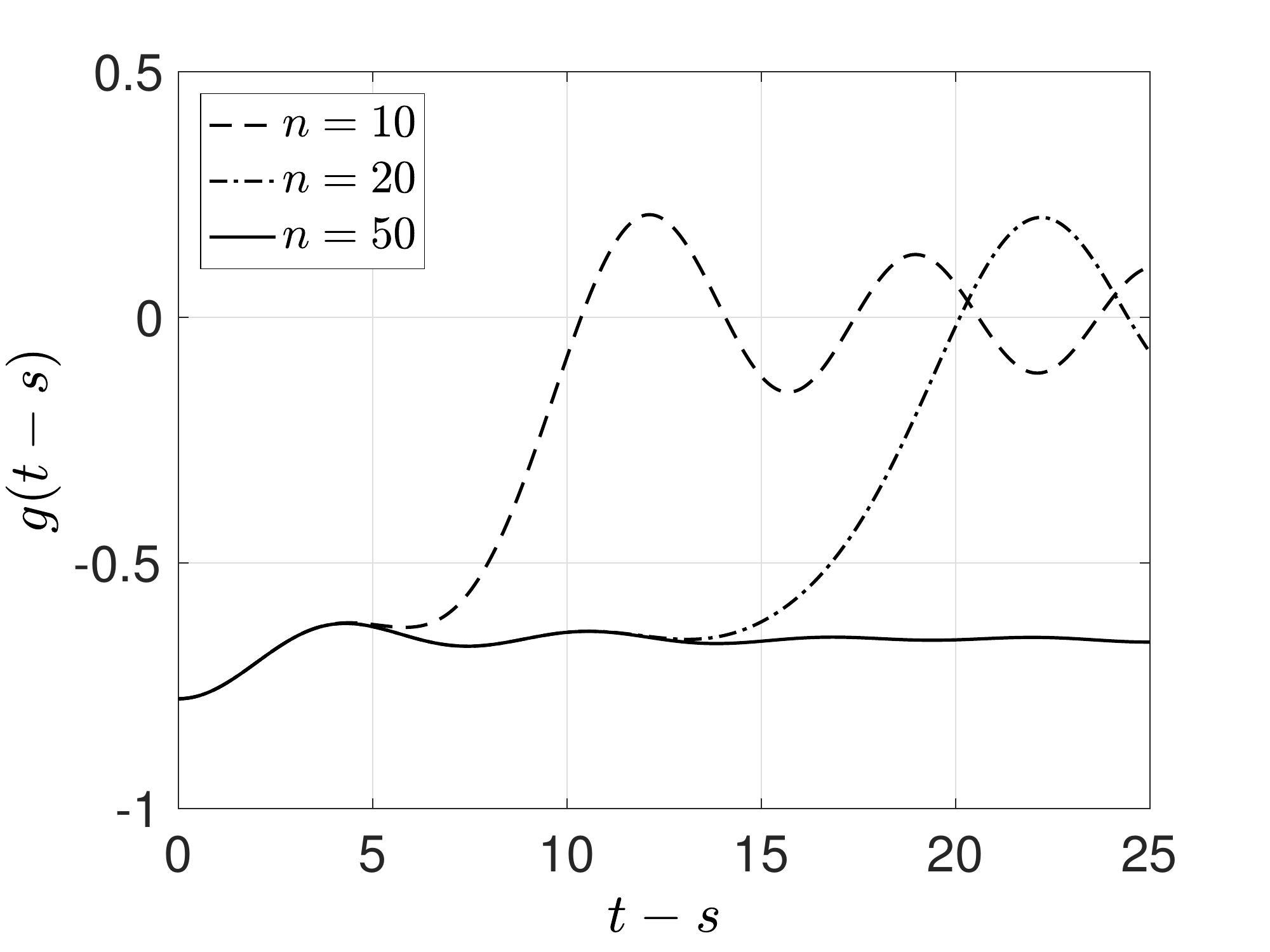} 
}
\caption{Dyson and Faber expansions of 
the Mori-Zwanzig memory kernel $g(t-s)$ in equation \eqref{MZFk}. 
Shown are results for different polynomial orders $n$. 
It is seen that the MZ-Faber series converges faster 
that the MZ-Dyson series.}
\label{fig:memory}
\end{figure}
In Figure \ref{fig:wave_MZ} we study the accuracy of the MZ-Dyson 
and the MZ-Faber expansions in representing the mean wave solution 
as a function of the polynomial order $n$. 
To this end, we solve \eqref{equ:mean_evolution} numerically 
a linear multi-step (explicit) time integration 
scheme (3rd-order Adams-Bashforth) combined with a 
trapezoidal rule to discretize the memory integral.
As easily seen, that the MZ-Faber expansion 
converges faster than the MZ-Dyson expansion.

\begin{figure}[t]
\centerline{Mean Wave Amplitude at $(r,\theta)=(1.1,0.1)$}
\centerline{
\includegraphics[height=6.3cm]{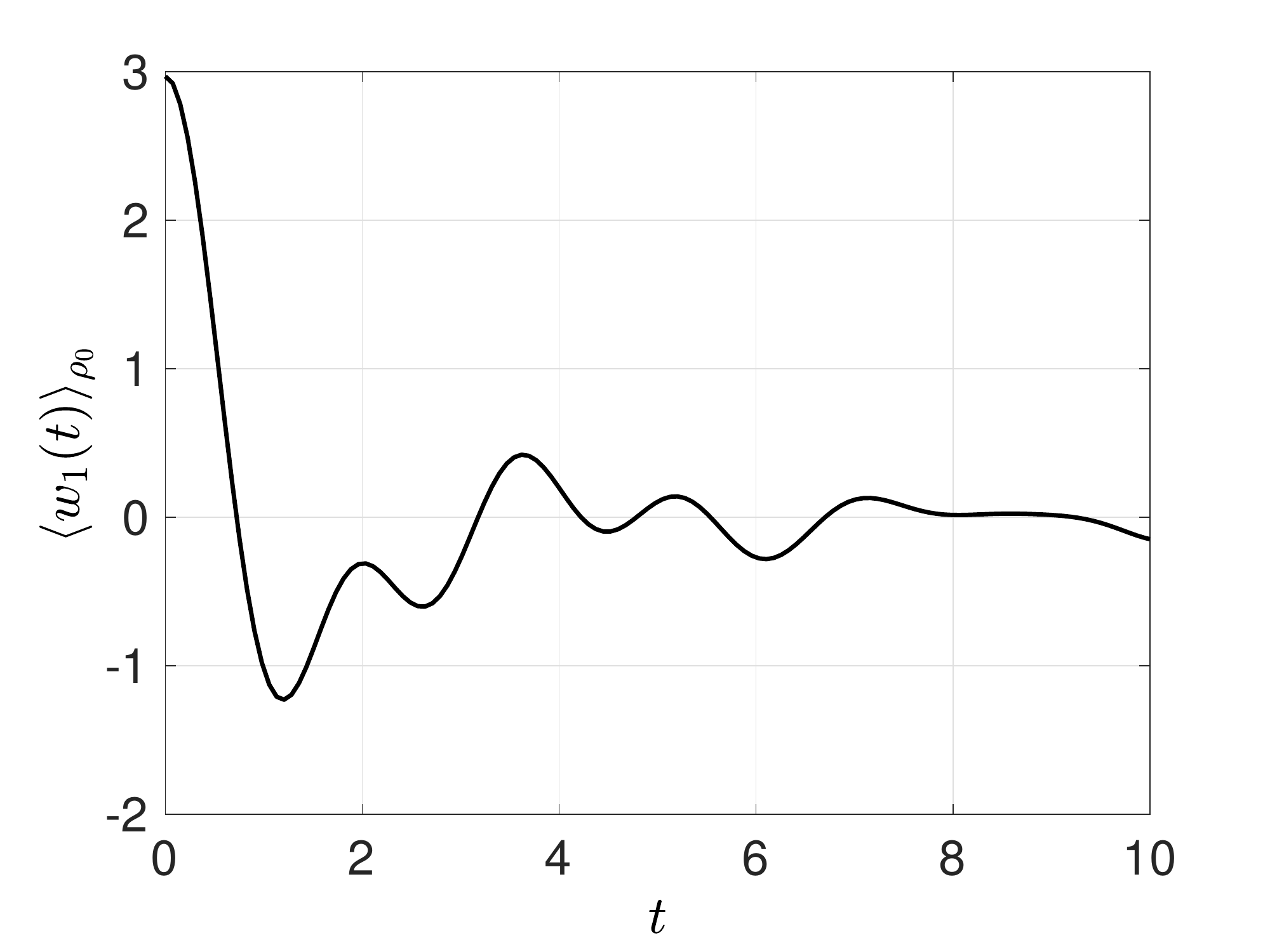}
}
\centerline{\hspace{0.5cm}MZ-Dyson Error\hspace{5.5cm}MZ-Faber Error}
\centerline{
\includegraphics[height=6cm]{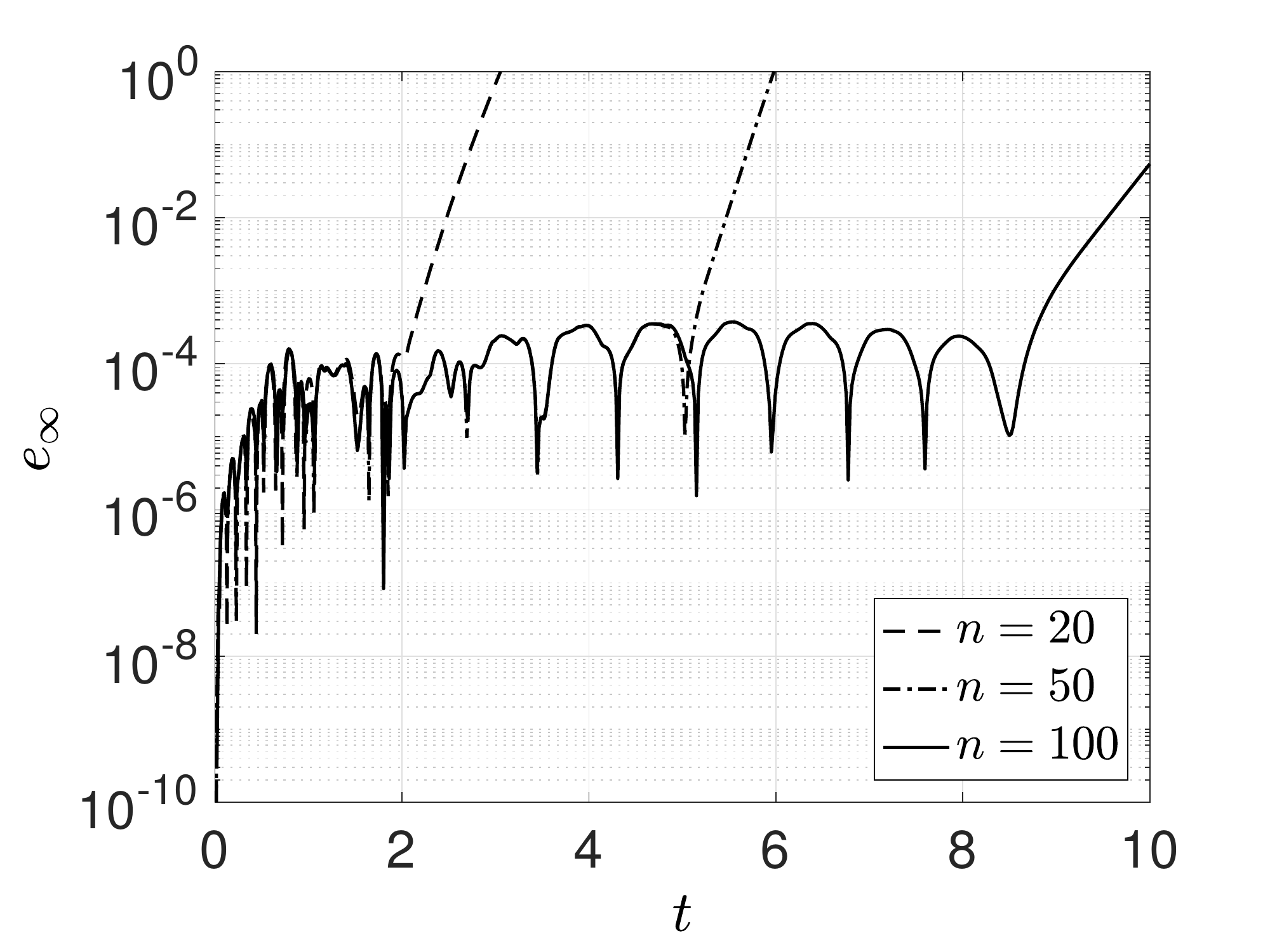} 
\includegraphics[height=6cm]{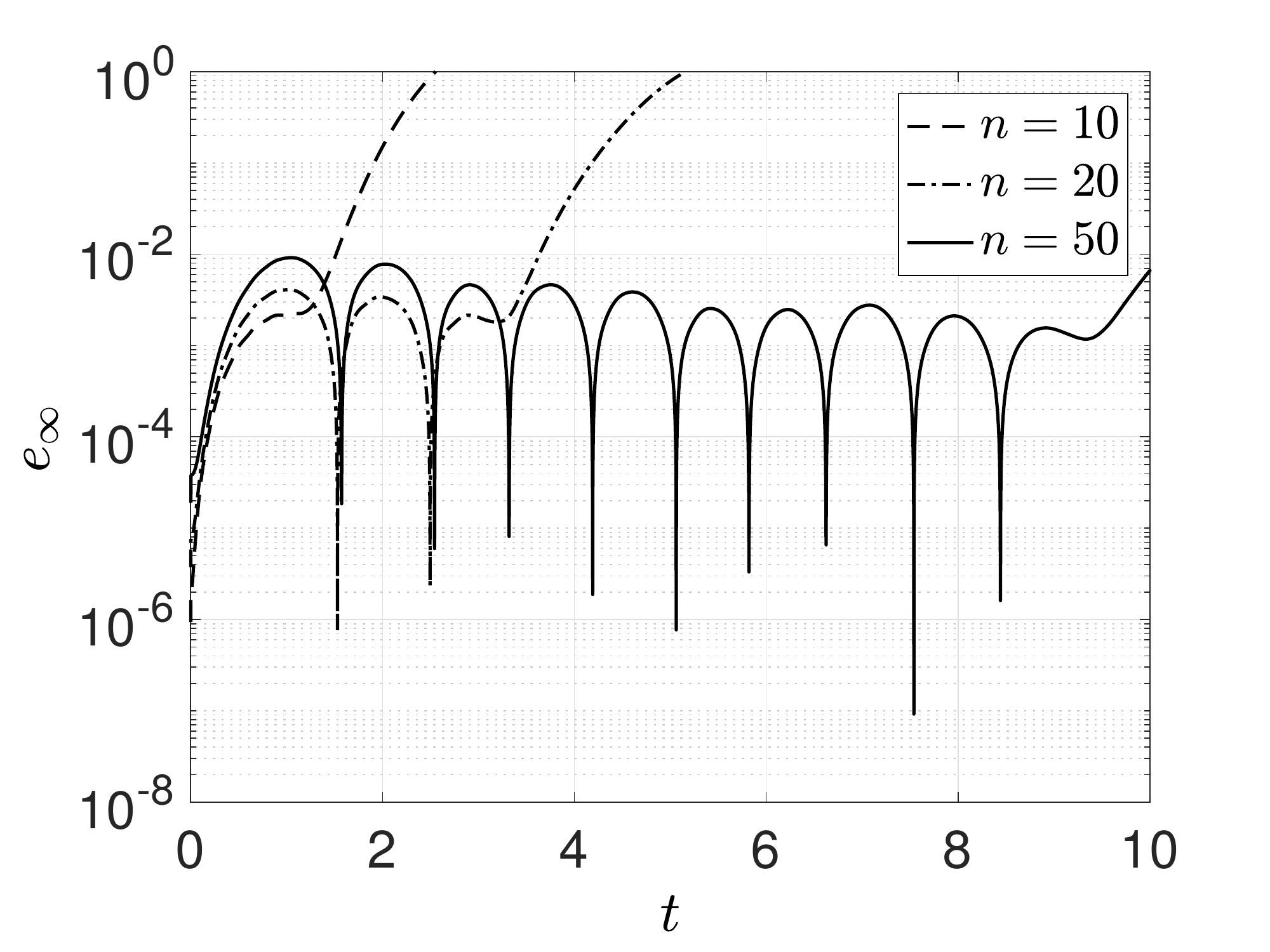}
}
\caption{MZ-Dyson and MZ-Faber approximation errors of the mean 
wave amplitude at $(r,\theta)=(1.1,0.1)$ as a function of the polynomial 
order $n$. It is seen that the MZ-Faber expansion converges faster 
than the MZ-Dyson series.}
\label{fig:wave_MZ}
\end{figure}

\subsection{Harmonic Chains on the Bethe Lattice}
\label{sec:chains}
Dynamics of harmonic chains on Bethe lattices is a 
simple but illustrative Hamiltonian dynamical 
system that has been widely studied in statistical mechanics, 
mostly in relation to Brownian motion  
\cite{baxter2016exactly,florencio1985exact,espanol1996dissipative,ford1965statistical,kim2000dynamics}. 
A Bethe lattice is a connected cycle-free graph in which each 
node interacts only with its neighbors. The number of such 
neighbors, is a constant of the graph called 
{\em coordination number}. This means that each node in the 
graph (with the exception of the leaf nodes) 
has the same number of edges connecting it to its neighbors.
In Figure \ref{fig:bethe_lattice} we show two Bethe lattices 
with coordination numbers $l=2$ and $l=3$, respectively.
\begin{figure}[H]
\centerline{
\includegraphics[height=5.5cm]{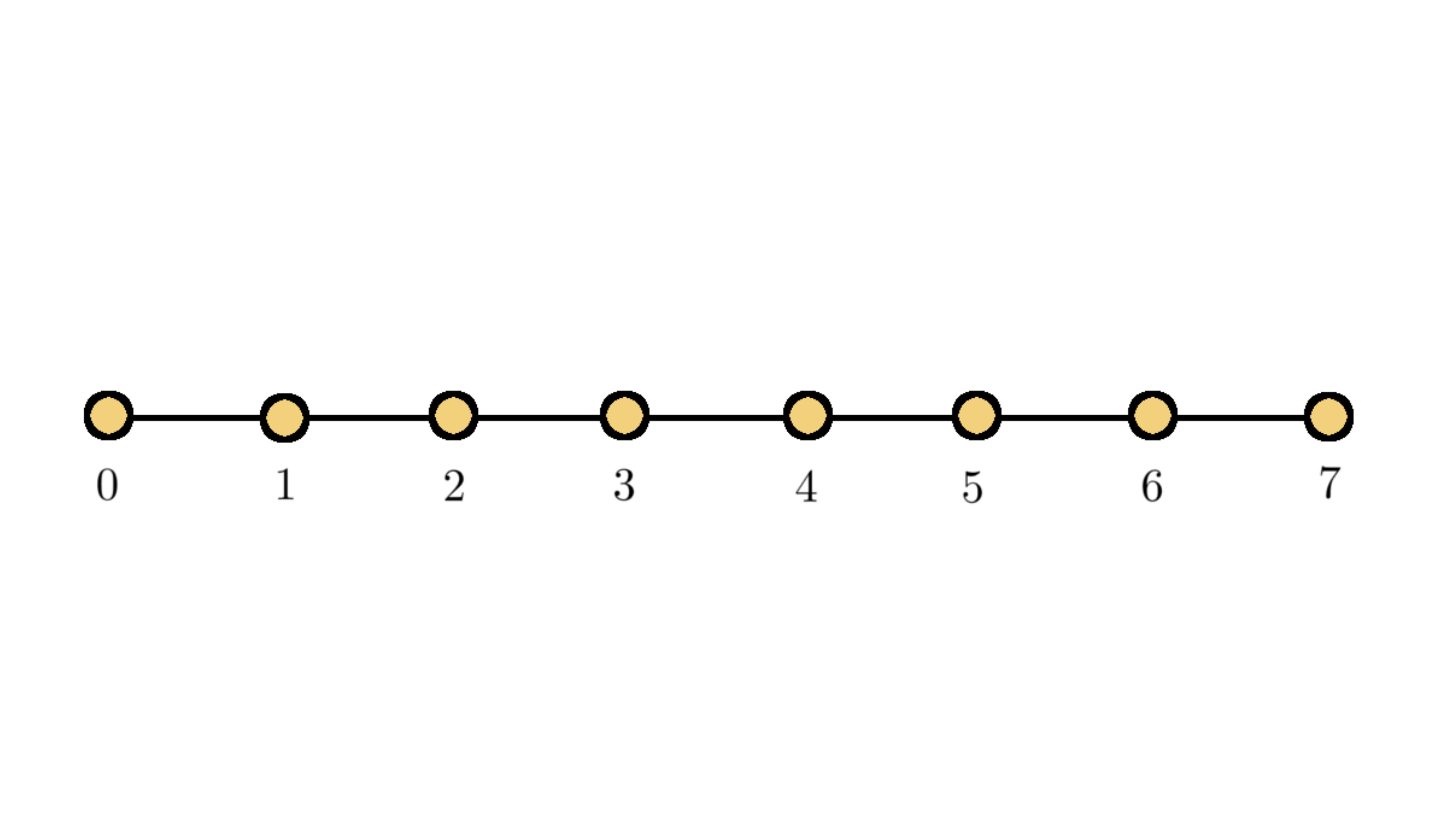}
\includegraphics[height=5.5cm]{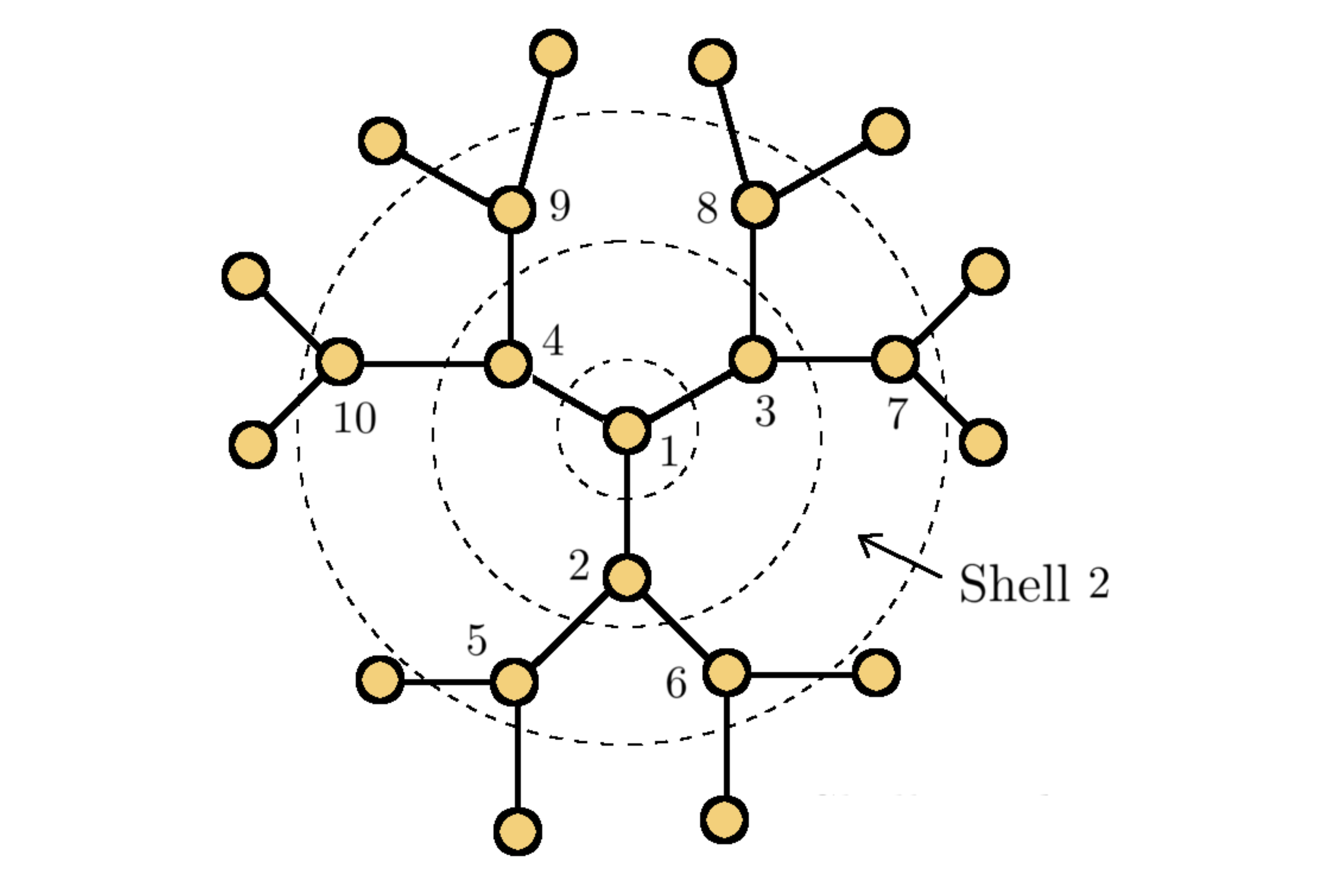}
}
\caption{Bethe lattices with coordination numbers $2$ (left), and $3$ (right).}
\label{fig:bethe_lattice}
\end{figure}
\noindent The Bethe graph is hierarchical and 
therefore it can be organized 
into shells, emanating from an arbitrary node. 
The number of nodes in the $k$-th shell is 
given by $N_k=l(l-1)^{k-1}$,while the 
total number of nodes within $S$ shells is
\begin{equation} 
 N=1+\sum_{k=1}^S N_k.
 \label{nodes}
 \end{equation}
Next, we consider a coupled system of $N$ harmonic 
oscillators\footnote{The number of oscillators cannot be 
set arbitrarily as it must satisfy the topological graph 
constraints prescribed by \eqref{nodes}.}
whose mutual interactions are defined 
by the adjacency matrix $\bm B^{(l)}$ of a Bethe graph with coordination number $l$  
\cite{biggs1993algebraic}.  
The Hamiltonian of such system can be written as   
\begin{align}
\label{bethe_hamiltonian}
H( \bm p,\bm q)=\frac{1}{2m}\sum_{i=1}^N p_i^2+
\frac{k}{2l}\sum_{i,j=1}^N B^{(l)}_{ij} (q_i-q_j)^2,
\end{align}
where $q_i$ and $p_i$ are, respectively, the displacement 
and momentum of the $i-$th particle, $m$ is the 
mass of the particles (assumed constant throughout the network), 
and $k$ is the elasticity constant that modulates the intensity of the 
quadratic interactions. 
We emphasize that the harmonic chain we consider here is  
one-dimensional. The Bethe graph basically just 
sets the interaction among the different oscillators.  
The dynamics of the harmonic chain on the Bethe lattice 
is governed by the Hamilton's equations
\begin{equation}
\frac{dq_i}{dt} = \frac{\partial H}{\partial p_i},\qquad
\frac{dp_i}{dt} = -\frac{\partial H}{\partial q_i}.
\end{equation}
These equations can be written in a matrix-vector form as
\begin{align}
\label{equ:bethe_l}
\left[
\begin{matrix}
\dot{\bm p}\\
\dot{\bm q}
\end{matrix}
\right]
=
\left[
\begin{matrix}
\bm 0&k \bm B^{(l)}-k\bm D^{(l)} \\
\bm I/m&\bm 0
\end{matrix}
\right]
\left[
\begin{matrix}
\bm p\\
\bm q
\end{matrix}
\right]=\bm C\left[
\begin{matrix}
\bm p\\
\bm q
\end{matrix}
\right],
\end{align}
where $\bm B^{(l)}$ is the adjacency matrix of the graph and $\bm D^{(l)}$ is the degree matrix. Note that \eqref{equ:bethe_l} is a linear dynamical system. 
The time evolution of any phase space function $u(\bm q,\bm p)$ 
(quantity of interest) satisfies 
\begin{align*}
\frac{du}{dt}=&\{u,H\},
\end{align*}
where
\begin{equation}
\{u,H\}=\sum_{i=1}^N\left(\frac{\partial u}{\partial q_i}\frac{\partial H}{\partial p_i}-\frac{\partial H}{\partial q_i}\frac{\partial u}{\partial p_i}\right)
\end{equation}
denotes the Poisson Bracket.
A particular phase space function we consider hereafter 
is the velocity auto-correlation function of a tagged 
oscillator, say the one at location $j=1$ 
(see Figure \ref{fig:bethe_lattice}). Such correlation 
function is defined as 
\begin{equation}
C_{p_1}(t)=\frac{\langle p_1(t)p_1(0)\rangle_{eq}}
{\langle p_1(0)p_1(0)\rangle_{eq}},
\label{autocorr}
\end{equation}
where the average is an integral over the 
Gibbs canonical distribution \eqref{gibbs}.

\subsubsection{Analytical Expressions for the Velocity Autocorrelation Function}
The simple structure of harmonic chains on 
the Bethe lattice allows us to determine analytical expressions for the 
velocity autocorrelation function \eqref{autocorr}, e.g., 
\cite{baxter2016exactly,kim2000dynamics,florencio1985exact}.

\paragraph{Bethe Lattice with Coordination Number 2}
Let us set $l=2$. In this case, 
the Bethe lattice is a a path graph, i.e., a 
one-dimensional chain of harmonic oscillators where each 
oscillator interacts only with the one at the left and at the 
right. We set fixed boundary conditions at the endpoint of the chain, 
i.e., $q_0(t)=q_{N+1}(t)=0$ and $p_0(t)=p_{N+1}(t)=0$
(particles are numbered from left to right). In this setting, 
the velocity auto-correlation function 
of the particle labeled with $j=1$ can be obtained analytically 
by employing Lee's continued fraction method  
\cite{florencio1985exact}. This yields the 
well-known  $J_0-J_{4}$ solution
\begin{align}
C_{p_1}(t)=J_0(2\omega t)-J_{4}(2\omega t),
\label{VCF_analytic}
\end{align}
where $J_{i}(t)$ is the $i$-th Bessel function of the 
first kind, and $\omega=k/m$. 
Here we choose $k=m=1$.
The Hamilton's equations \eqref{equ:bethe_l} for the inner oscillators\footnote{We  
exclude the two oscillators at the endpoints of the harmonic chain, since their dynamics is trivial.} take the form 
\begin{align}
\label{equ:bethe_l2}
\left[
\begin{matrix}
\dot{\bm p}\\
\dot{\bm q}
\end{matrix}
\right]
=
\left[
\begin{matrix}
\bm 0& \bm B^{(2)}-\bm D^{(2)}\\
\bm I &\bm 0
\end{matrix}
\right]
\left[
\begin{matrix}
\bm p\\
\bm q
\end{matrix}
\right],
\end{align}
where  $\bm B^{(2)}$ and $\bm D^{(2)}$ are the adjacency 
matrix and the degree matrix of the Bethe lattice with $l=2$ (see Figure \ref{fig:bethe_lattice}). 
As an example, if we consider five oscillators 
then $\bm B^{(2)}$ and $\bm D^{(2)}$ are given by  
\begin{equation}
\bm B^{(2)}=
\left[
\begin{array}{ccc}
0 & 1 & 0\\
1 & 0 & 1\\
0 & 1 & 0
\end{array}
\right],
\qquad 
\bm D^{(2)}=
\left[
\begin{array}{ccc}
2 & 0 & 0\\
0 & 2 & 0\\
0 & 0 & 2
\end{array}
\right].
\end{equation}

\begin{figure}[t]
\centerline{
\centerline{Bethe lattice with $l=2$\hspace{4cm}
Bethe lattice  with $l=3$}}
\centerline{
\includegraphics[height=6cm]{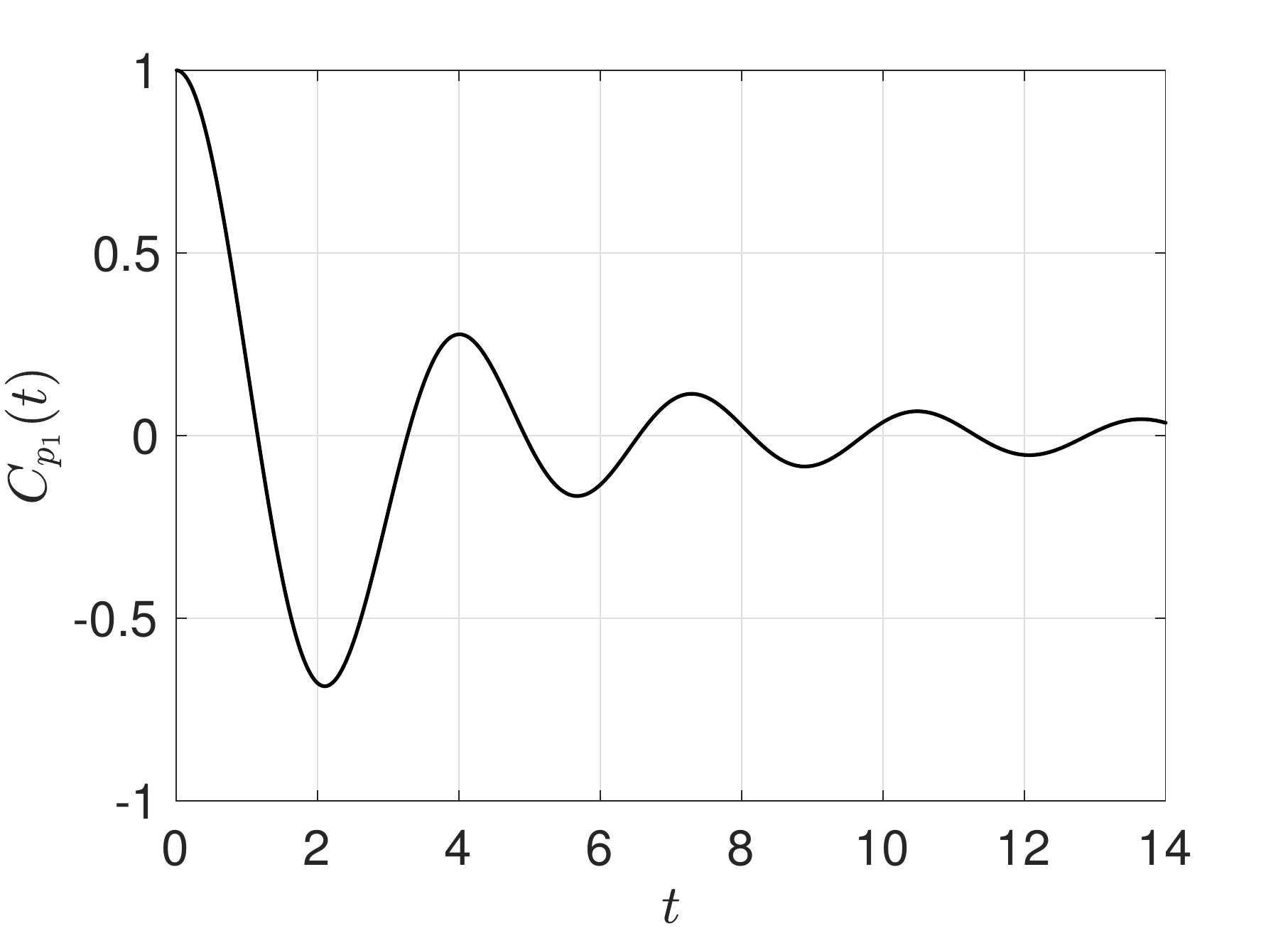} 
\includegraphics[height=6cm]{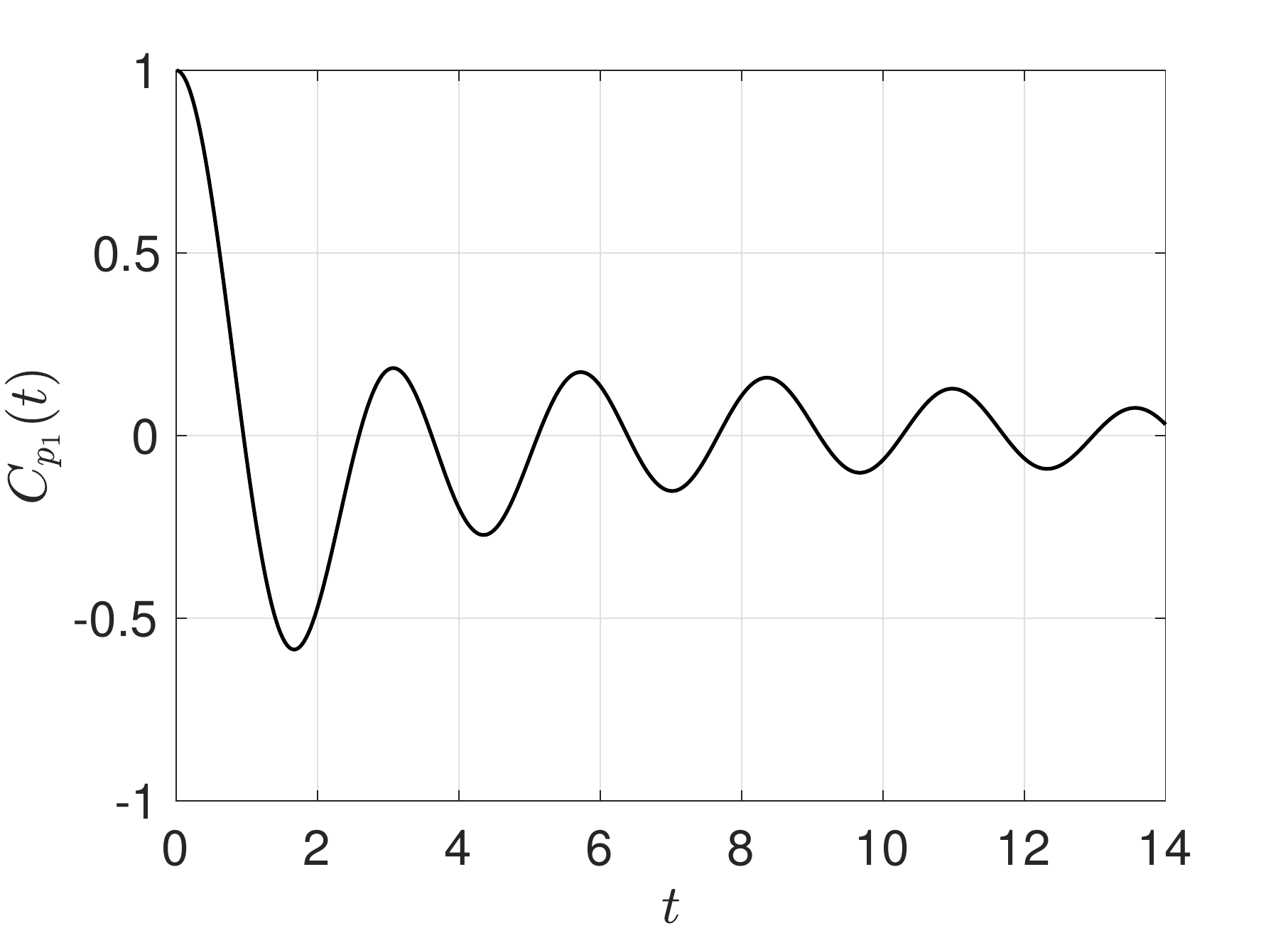} 
}
\caption{Velocity auto-correlation functions 
\eqref{VCF_analytic} (left) and \eqref{VCF_l2} (right) 
of a tagged oscillator in an harmonic chain interacting 
on a Bethe lattice with coordination number $l=2$ and $l=3$, 
respectively.  
}
\label{fig:VCF} 
\end{figure}
\paragraph{Bethe Lattice with Coordination Number 3} 
Bethe graphs with $l=3$ can be represented as planar graphs
(see Figure \ref{fig:bethe_lattice}). The velocity 
auto-correlation function at the center node can 
be expressed analytically \cite{kim2000dynamics}, 
in the limit of an infinite number of oscillators ($N\rightarrow \infty$)\footnote{Thanks to the symmetry of the Bethe lattice, 
in the limit $n\rightarrow \infty$ and with free boundary conditions 
the velocity auto-correlation function is the same at each node.}, 
as 
\begin{align}
C_{p_1}(t)&=\sum_{n=-\infty}^{+\infty}[G_n(t)+H_n(t)]J_{2n}(bt)
\label{VCF_l2}
\end{align}
where 
\begin{align*}
G_n(l)&=\sum_{k=0}^{\infty}\frac{g_k(l)}{b^{2k-2}}\frac{1}{2\pi}\int_a^{\pi/2}d\theta \frac{\cos^2(\theta)}{\sin^{2k}(\theta)}\cos(2n\theta),\quad g_k(l)=-\sum_{j=k}^{\infty}\frac{(2j-1)!!}{[2^j(2j-1)j!]}a^{2j}c^{2(k-j)},\\
H_n(l)&=\sum_{k=0}^{\infty}\frac{h_k(l)}{b^{-2k-2}}\frac{1}{2\pi}\int_a^{\pi/2}d\theta \frac{\cos^2(\theta)}{\sin^{-2k}(\theta)}\cos(2n\theta),\quad h_k(l)=-\sum_{j=k}^{\infty}\frac{(2j-1)!!}{[2^j(2j-1)j!]}a^{2(j-k)}c^{-2j}
\end{align*}
and $a=\sqrt{2}-1$, $b=\sqrt{2}+1$ and $c=\sqrt{6}$. 
The  Hamilton's equations  of motion in this case are\footnote{Here we implemented a free boundary condition at the outer shell of the chain.} ($k=m=1$)
\begin{align}
\label{equ:bethe_l3}
\left[
\begin{matrix}
\dot{\bm p}\\
\dot{\bm q}
\end{matrix}
\right]
=
\left[
\begin{matrix}
\bm 0&\bm B^{(3)}-\bm D^{(3)}\\
\bm I&\bm 0
\end{matrix}
\right]
\left[
\begin{matrix}
\bm p\\
\bm q
\end{matrix}
\right]
.
\end{align}
where  $\bm B^{(3)}, \bm D^{(3)}$ are the adjacency matrix 
and the degree matrix of the Bethe lattice with $l=3$ (see Figure \ref{fig:bethe_lattice}). For example,
if we label the oscillators as in Figure \ref{fig:bethe_lattice}, and 
assume that the Bethe lattice has only three shells, i.e., $10$ 
oscillators ($4$ inner nodes, and $6$ leaf nodes) then the 
adjacency matrix and the degree matrix are 
\begin{equation}
\bm B^{(3)}=
\left[
\begin{array}{cccccccccc}
0 & 1 & 1 & 1 & 0 & 0 & 0 & 0 & 0 & 0 \\
1 & 0 & 0 & 0 & 1 & 1 & 0 & 0 & 0 & 0 \\
1 & 0 & 0 & 0 & 0 & 0 & 1 & 1 & 0 & 0 \\
1 & 0 & 0 & 0 & 0 & 0 & 0 & 0 & 1 & 1 \\
0 & 1 & 0 & 0 & 0 & 0 & 0 & 0 & 0 & 0 \\
0 & 1 & 0 & 0 & 0 & 0 & 0 & 0 & 0 & 0 \\
0 & 0 & 1 & 0 & 0 & 0 & 0 & 0 & 0 & 0 \\
0 & 0 & 1 & 0 & 0 & 0 & 0 & 0 & 0 & 0 \\
0 & 0 & 0 & 1 & 0 & 0 & 0 & 0 & 0 & 0 \\
0 & 0 & 0 & 1 & 0 & 0 & 0 & 0 & 0 & 0 
\end{array}
\right],\qquad
\bm D^{(3)}=
\left[
\begin{array}{cccccccccc}
3 & 0 & 0 & 0 & 0 & 0 & 0 & 0 & 0 & 0 \\
0 & 3 & 0 & 0 & 0 & 0 & 0 & 0 & 0 & 0 \\
0 & 0 & 3 & 0 & 0 & 0 & 0 & 0 & 0 & 0 \\
0 & 0 & 0 & 3 & 0 & 0 & 0 & 0 & 0 & 0 \\
0 & 0 & 0 & 0 & 1 & 0 & 0 & 0 & 0 & 0 \\
0 & 0 & 0 & 0 & 0 & 1 & 0 & 0 & 0 & 0 \\
0 & 0 & 0 & 0 & 0 & 0 & 1 & 0 & 0 & 0 \\
0 & 0 & 0 & 0 & 0 & 0 & 0 & 1 & 0 & 0 \\
0 & 0 & 0 & 0 & 0 & 0 & 0 & 0 & 1 & 0 \\
0 & 0 & 0 & 0 & 0 & 0 & 0 & 0 & 0 & 1 
\end{array}
\right].
\end{equation}

\subsubsection{Generalied Langevin Equation for the Velocity Autocorrelation Function}
The evolution equation for the velocity autocorrelation 
function \eqref{autocorr} was obtained in Section \ref{sec:autocorrelation} 
and it is hereafter rewritten for convenience
\begin{equation}
\frac{d C_{p_1}(t) }{d t}=  a C_{p_1}(t)+\int_0^tg(t-s)C_{p_1}(s)ds.
\label{equ:evolution_correlation1}
\end{equation}
The initial condition is  $C_{p_i}(0)=1$.  
The MZ-Dyson and MZ-Faber series expansions of the  
the memory kernel $g(t-s)$ are given by 
\begin{align*}
g(t-s) = \sum_{j=0}^n \frac{g^D_j}{j!}(t-s)^j,\qquad 
g(t-s) = \sum_{j=0}^n g^F_j e^{tc_0} \frac{J_j(2t\sqrt{-c_1})}{(\sqrt{-c_1})^j} 
\end{align*}
where 
\begin{align*}
g_j^D=\bm b^T (\Mi_{11}^T)^j \bm a,
\quad 
g_j^F=\bm b^{T}\F_j\left(\bm M_{11}^{T}\right)\bm a.
\end{align*}
The definition of the matrix $\bm M_{11}^T$ and the vectors 
$\bm a$, $\bm b$ is the same as before. 
Here we used the fact that for any quadratic Hamiltonian we have $\left<p_i(0),q_i(0)\right>_{eq}=0$ and 
$\left<p_i(0),p_j(0)\right>_{eq}=\delta_{ij}$.
 In Figure \ref{fig:memory_chain} we study convergence 
 of the MZ-Dyson and the  MZ-Faber series expansion   
of the memory kernel in equation \eqref{equ:evolution_correlation}. 
As before, the MZ-Faber series converges  faster 
that the MZ-Dyson series. 
\begin{figure}[t]
\centerline{\hspace{0.5cm}MZ-Dyson\hspace{6.5cm}MZ-Faber}
\centerline{
\includegraphics[height=6cm]{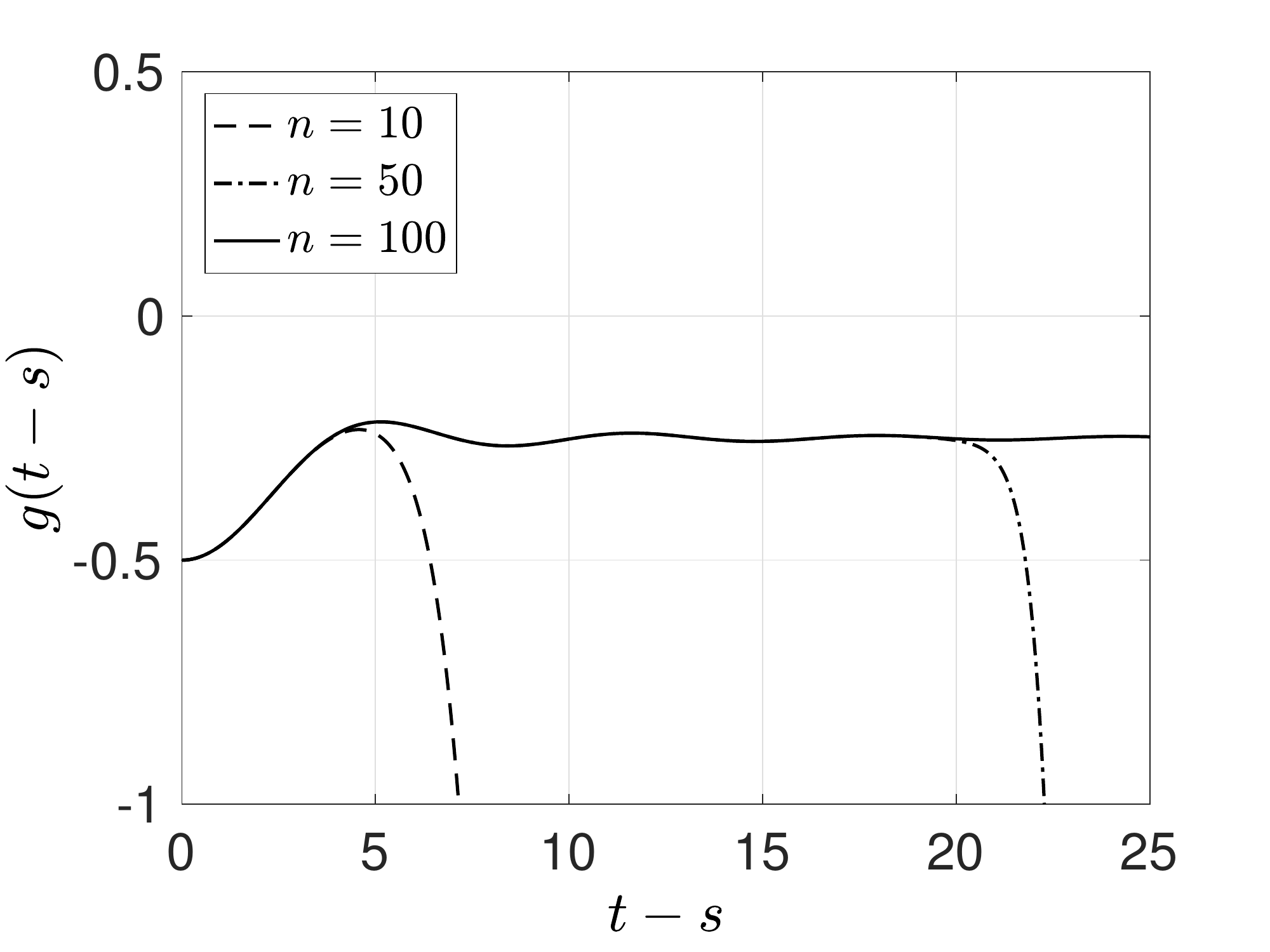} 
\includegraphics[height=6cm]{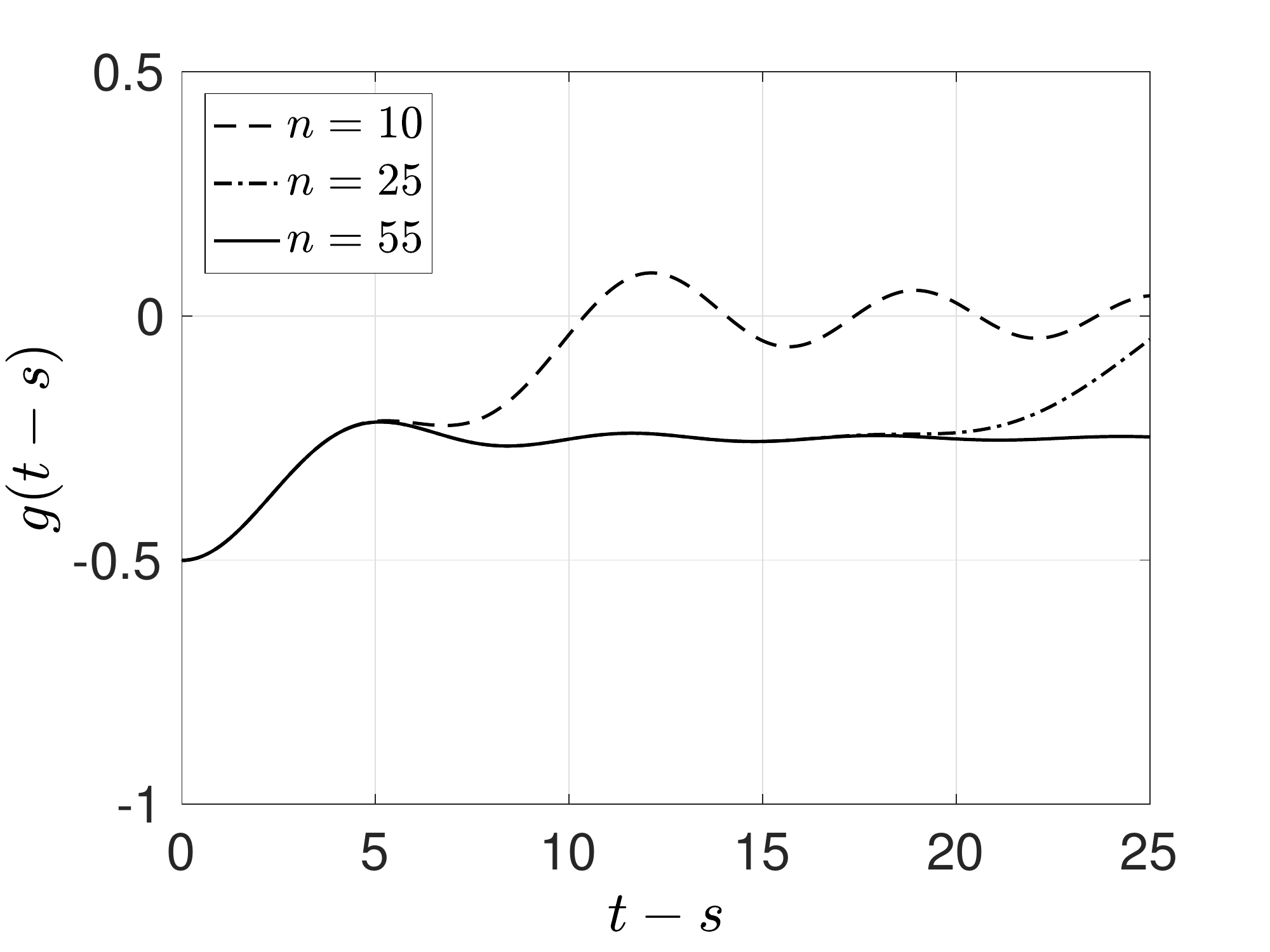} 
}
\caption{Harmonic chains of oscillators. Dyson and Faber expansions of 
the Mori-Zwanzig memory kernel $g(t-s)$. 
Shown are results for different polynomial orders $n$. 
It is seen that the MZ-Faber series converges  faster 
that the MZ-Dyson series.}
\label{fig:memory_chain}
\end{figure}
In Figure \ref{fig:l2_conv} and Figure \ref{fig:l3_conv}, 
we study the accuracy of the MZ-Dyson and the MZ-Faber expansions
in representing the velocity auto-correlation 
functions \eqref{VCF_analytic} and \eqref{VCF_l2}
(see Figure \ref{fig:VCF}).
Specifically, in these simulations we considered a chain 
of $N=100$ oscillators for the case $l=2$, 
and $8$ shells of oscillators for the case $l=3$, i.e., 
a total number of $N=766$ oscillators.
The results in Figure  \ref{fig:l2_conv} and Figure \ref{fig:l3_conv} 
show that both the MZ-Dyson  
and the MZ-Faber expansions of the memory 
integral yield accurate approximations of the 
velocity autocorrelation function, and that 
convergence is uniform with the polynomial order.  
\begin{figure}[t]
\centerline{
\includegraphics[height=6.5cm]{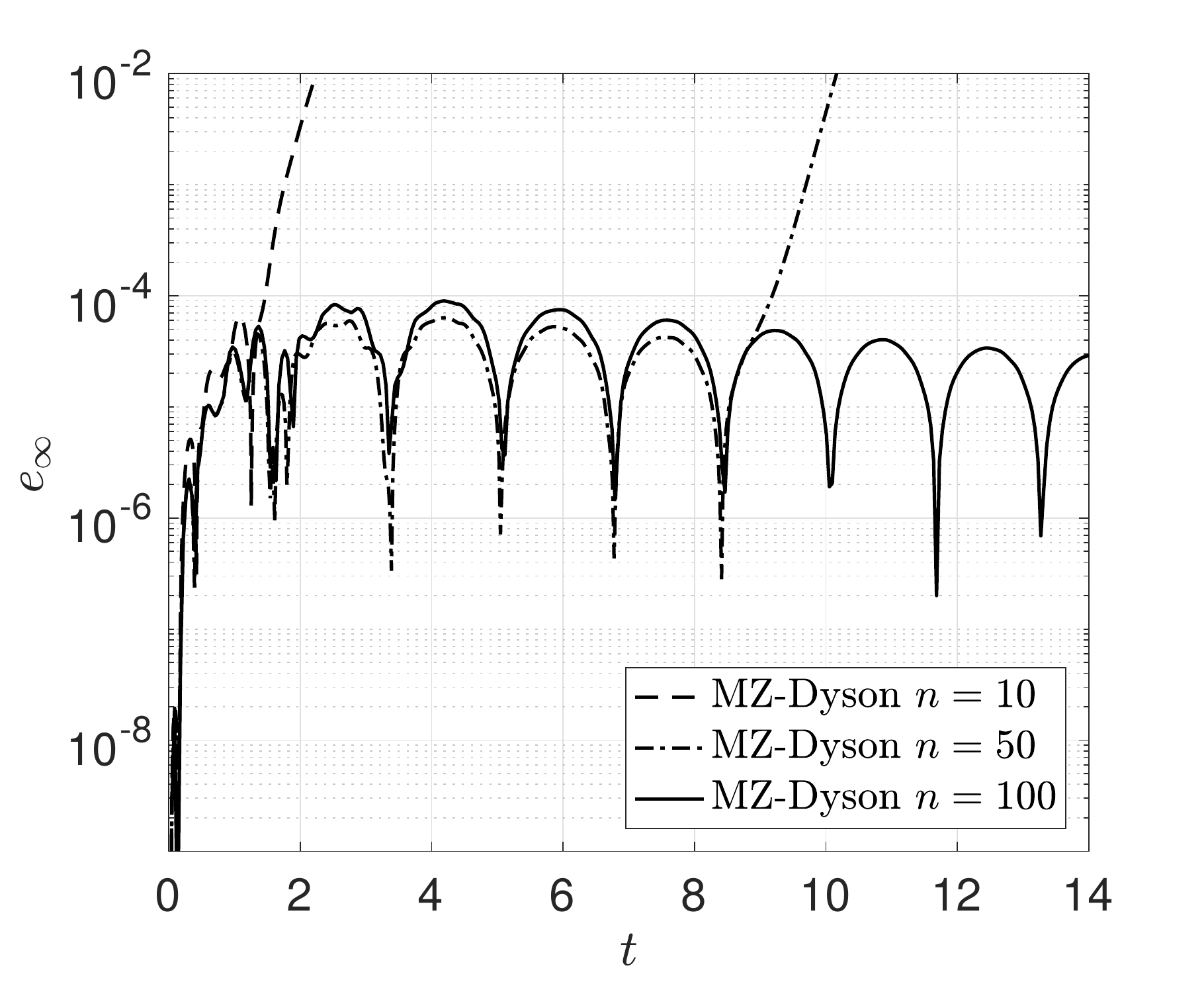} 
\includegraphics[height=6.5cm]{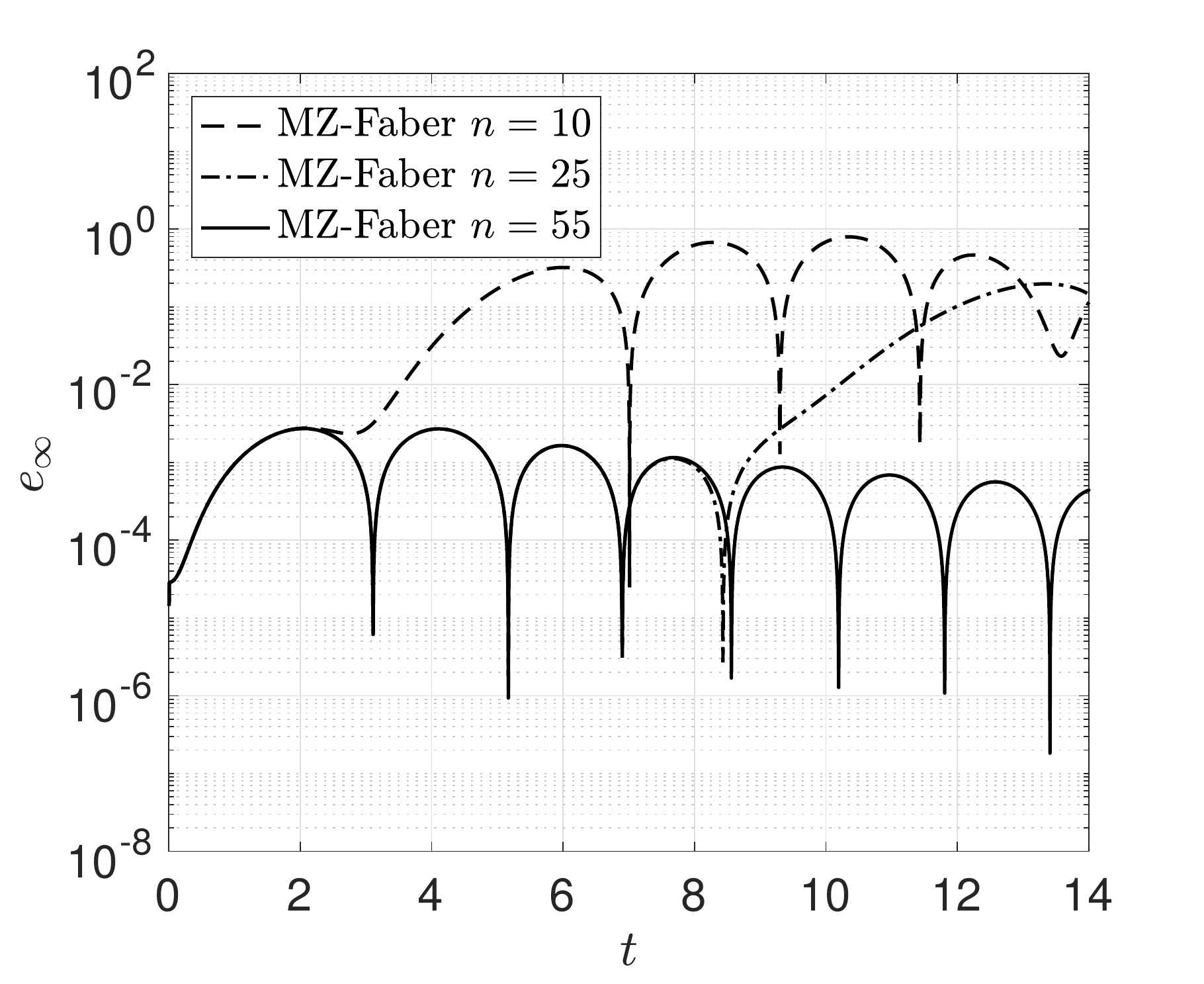} 
}
\caption{Accuracy of the MZ-Dyson and MZ-Faber expansions in 
representing the velocity auto-correlation function of the 
tagged oscillator $j=2$ in an harmonic chain interacting
on the Bethe lattice with coordination number $2$. 
It is seen that the MZ-Dyson and the MZ-Faber 
expansions yield accurate predictions as 
we increase the polynomial order $n$. Moreover, the 
MZ-Faber expansion converges faster 
than the MZ-Dyson expansion.}
\label{fig:l2_conv} 
\end{figure}
\begin{figure}[t]
\centerline{ 
\includegraphics[height=6.5cm]{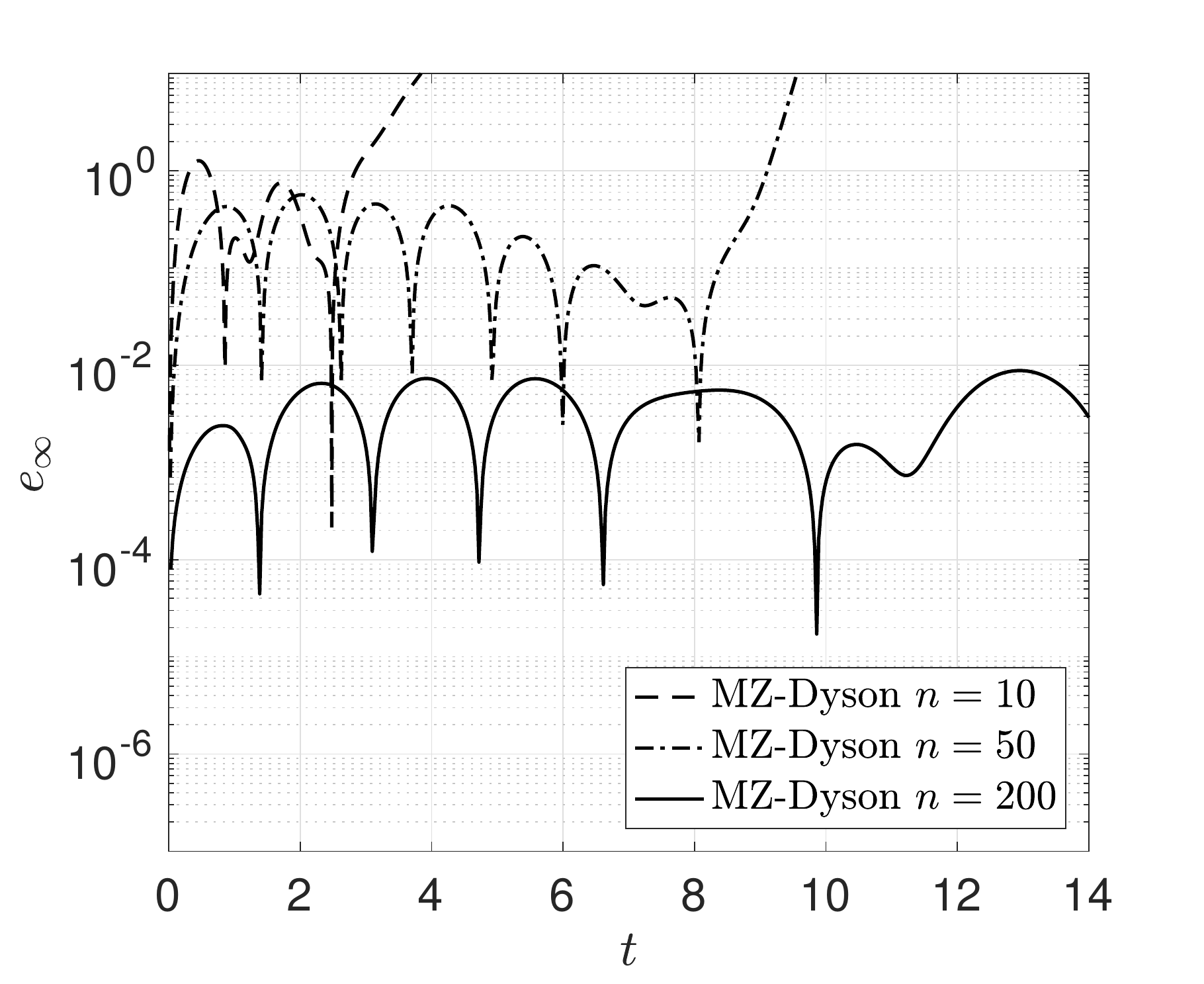} 
\includegraphics[height=6.5cm]{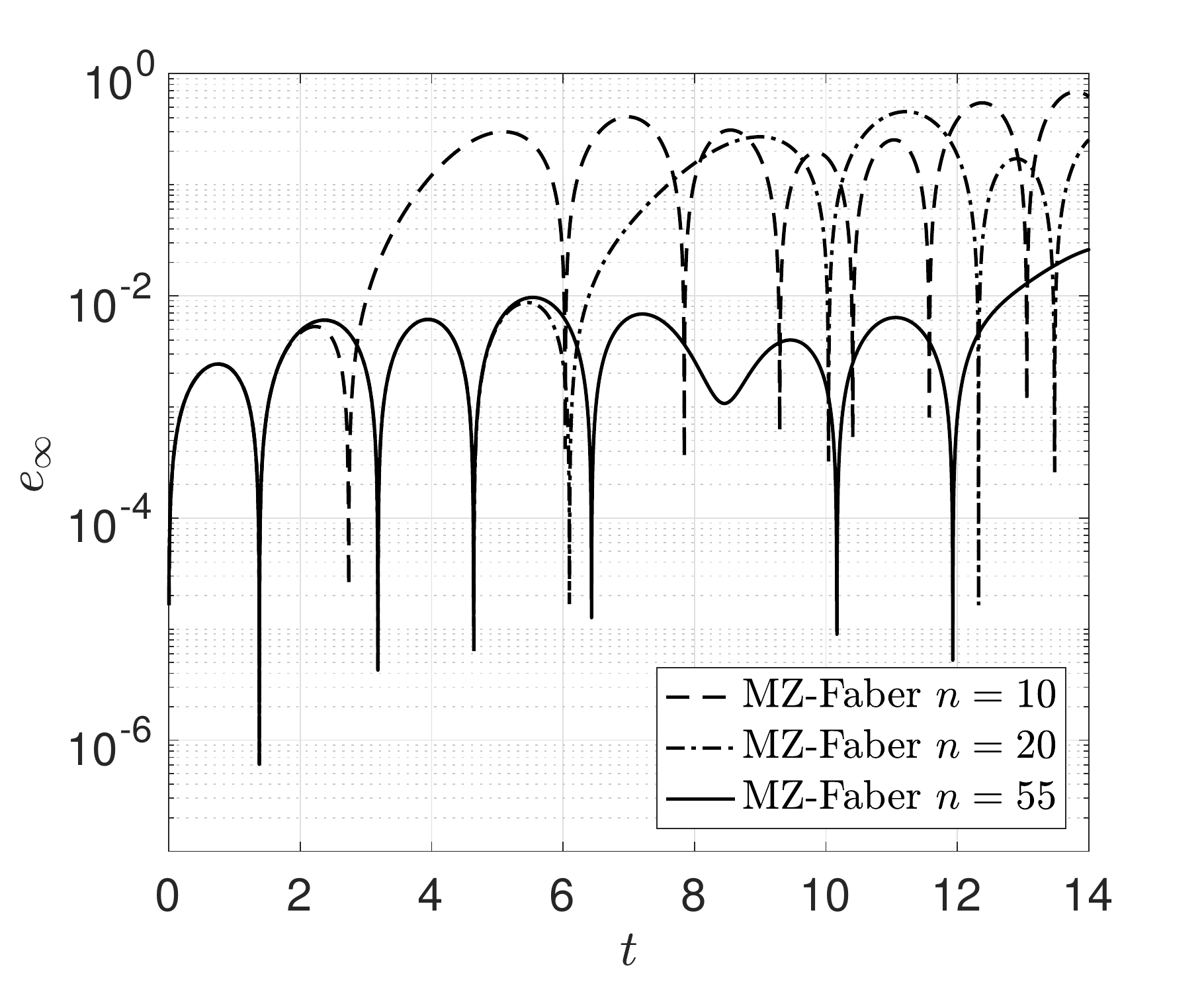} 
}
\caption{Accuracy of the MZ-Dyson and MZ-Faber expansions 
in representing the velocity auto-correlation function 
of the oscillator at the center of a Bethe lattice with 
coordination number $3$, $8$ shells and $N=766$ oscillators. 
It is seen that the MZ-Dyson and the MZ-Faber expansion yield accurate predictions as we increase the polynomial order $n$. Moreover, 
the MZ-Faber expansion converges faster than the MZ-Dyson series.}
\label{fig:l3_conv} 
\end{figure}
We emphasize that the new expansion of the MZ memory integral 
we developed can be employed to calculate 
phase space functions of harmonic oscillators on 
graphs with arbitrary topological structure. 
The following example shows the effectiveness 
of the proposed technique in calculating the 
velocity auto-correlation function of a tagged 
oscillator in a network sampled from the 
 Erd\"{o}s–R\'enyi random graph.

\subsubsection{Harmonic Chains on Graphs with Arbitrary Topology}
In this section we consider an harmonic chain on 
a graph with arbitrary topology. The Hamiltonian function is  
\begin{align}
H=\frac{1}{2m}\sum_{i=1}^N p_i^2+\frac{k}{2}\sum_{\substack{i,j=1\\i<j}}^N J_{ij}(q_i-q_j)^2,
\label{HR}
\end{align}
where $J_{ij}$ is here is assumed is to be 
a sample from the Erd\"{o}s–R\'{e}nyi random 
adjacency matrix \cite{Bollobas,Newman}. 
In Figure \ref{fig:r_n} we study the accuracy of the
MZ-Dyson and MZ-Faber expansions in approximating the 
velocity auto-correlation function of a tagged oscillator. 
In this case, no analytical solution is available 
and therefore we compared our solution 
to an accurate Monte Carlo benchmark. 
The lack of symmetry in each realization of the 
random network makes the velocity auto-correlation function 
dependent on the particular oscillator we consider. 

\begin{figure}[t]
\centerline{
\includegraphics[height=6.0cm]{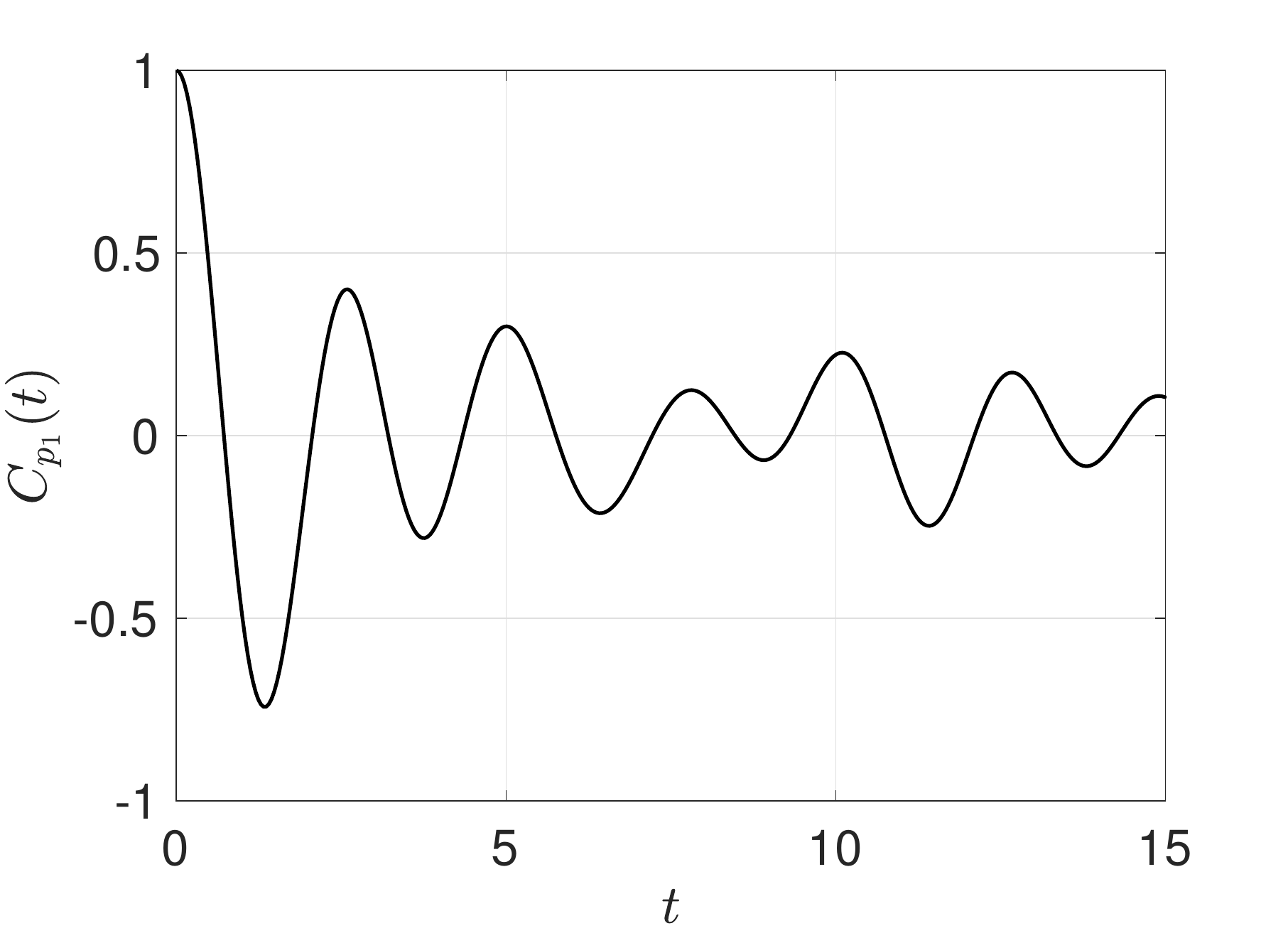} 
\includegraphics[height=6.0cm]{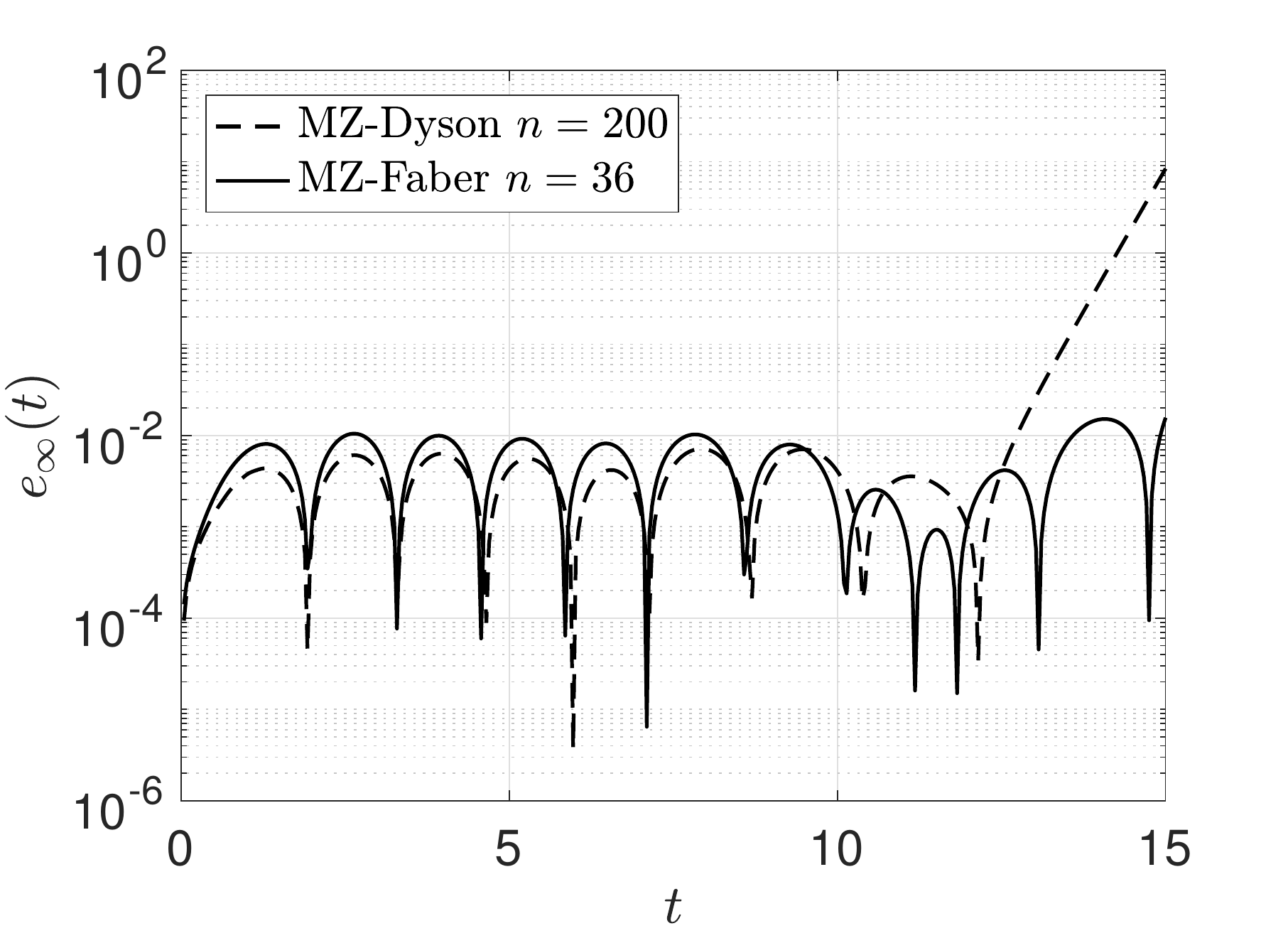} 
}
\caption{Aaccuracy of the MZ-Dyson and MZ-Faber expansions in 
approximating the velocity auto-correlation function of one tagged 
oscillator on a network obtained by sampling 
the Erd\"{o}s-Rényi graph $G(100,0.1)$. The benchmark 
solution is computed by  Monte Carlo simulation. 
}
\label{fig:r_n}
\end{figure}

\subsection*{Acknowledgements}
This work was supported by the Air Force Office of Scientific 
Research (AFOSR) grant  FA9550-16-586-1-0092.

\appendix 
\section{Faber Polynomials}
\label{app:Faber}
In this appendix we briefly review the theory of 
Faber polynomials in the complex plane. Such polynomial 
were introduced by Faber in \cite{Faber} (see \cite{Suetin} for 
a through review), and they play an important role 
in theory of univalent functions and in the approximation 
of matrix functions \cite{moret2001computation,novati2003solving}. 
To introduce Faber polynomials, let 
\begin{align*}
M=\left\{\Omega\subset \mathbb{C}:\, 
\textrm{$\Omega\neq\{\emptyset\}$ is compact and 
$\mathbb{C}\setminus\Omega$ is simply connected}\right\},
\end{align*}
Given any set $\Omega\subset M$, by the Riemann mapping 
theorem there exits a conformal surjection
\begin{align}
\psi:\hat{\mathbb{C}}\setminus\{w:|w|\leq \gamma\}\rightarrow \hat{\mathbb{C}}\setminus\Omega,\quad \psi(\infty)=\infty, \quad \psi'(\infty)=1,
\label{mapping}
\end{align}
where $\hat{\mathbb{C}}$ is the Riemann sphere. 
The constant $\gamma$ is called {\em capacity} of $\Omega$. 
The $j-$th order Faber polynomial $\F_j(z)$ is defined to be the regular  
part of the Laurent expansion of $[\psi^{-1}(z)]^j$ at  infinity, i.e., 
\begin{align}\label{faber_coeff}
\F_j(z):=z^j+\sum_{k=0}^{j-1}\beta_{j,k}z^k,\quad j\geq 0.
\end{align}
Let $\Gamma$ be the boundary of $\Omega$. For $R\geq \gamma$ 
we define the equipotential curve $\Gamma(R)$ as 
\begin{align}
\Gamma(R):=\{z:\psi^{-1}(z)=R\}.
\label{equipotential}
\end{align}
We also denote as $\Omega(R)$ the closure of the interior of $\Gamma(R)$. 
Obviously, if  $R =\gamma$ then we have $\Omega(R)=\Omega$  
and $\Gamma(R)=R$.
Any analytic function $f(z)$ on $\Omega$ can be uniquely expanded in terms 
of Faber polynomials as 
\begin{equation}
f(z)=\lim_{m\rightarrow \infty} f_m(z)\qquad f_m(z)=\sum_{j=0}^m a_j(f)\F_j(z), 
\label{truncated_Faber}
\end{equation} 
where the coefficients $a_j(f)$ are given be the complex integral 
\begin{equation}
a_j(f) = \frac{1}{2\pi i}\int_{|w|=R}\frac{f\left(\psi(w)\right)}{w^{j+1}}dw.
\label{a_jf}
\end{equation}
It can be shown that $\F_j(z)$  satisfy the following recurrence relation 
\begin{align}
\label{recursive}
\F_0(z)&=1,\nonumber\\
\F_1(z)&=z-c_0,\nonumber\\
&\vdots \nonumber\\
\F_j(z)&=(z-c_0)\F_{j-1}(z)-
(c_1\F_{j-2}(z)+...+c_{j-1}\F_0(z))-(j-1)c_{j-1},\quad j\geq2,
\end{align}
where $c_0,c_1,...$ are the coefficients of the Laurent 
series expansion of the mapping $\psi$, i.e., 
\begin{align}\label{laurent}
\psi(w)=w+c_0+\frac{c_1}{w}+ \frac{c_2}{w^2}+\cdots , \qquad |w|> \gamma
\end{align}
From a computational viewpoint, it is convenient to limit the number 
of terms in the expansion \eqref{laurent}. In this way, we can 
simplify the recurrence relation \eqref{recursive}, the 
calculation of \eqref{a_jf} and therefore significantly 
speed up computations. In this paper we consider the map
\begin{align}\label{circle_ellipse}
\psi(w)=w+c_0+\frac{c_1}{w}, 
\end{align}
which transforms circles into ellipses. 
In this case, the coefficients \eqref{a_jf} can be obtained 
analytically by computing the integral 
\begin{align}
a_j(t)&=\frac{1}{2\pi i}\int_{|w|=R}\frac{\exp\{t(w+c_0+c_1/w)\}}{w^{j+1}} dw, \nonumber \\
&=\frac{1}{(\sqrt{-c_1})^{j}} e^{tc_0}J_j(2t\sqrt{-c_1}),
\label{a_j_analytic}
\end{align}
where $J_j(x)$ is the Bessel function of the first kind. 
The number of terms in the Laurent series 
expansion \eqref{laurent} 
should be selected so that the spectrum of 
the operator $\Q\L$ lies entirely within the 
equipotential curve \eqref{equipotential}. 
In the numerical examples we discuss in 
Section \ref{sec:application} such spectrum 
turns out to be relatively concentrated around 
the imaginary axis. Hence, the second-order 
truncation \eqref{circle_ellipse}, which defines an 
elliptical equipotential curve, guarantees 
fast convergence of the Faber series expansion of 
the orthogonal dynamics propagator. 

\section{Faber Expansion of the Orthogonal Dynamics Propagator}
Given any matrix representation of the operator $\Q\L$ (generator 
of the orthogonal dynamics) and a vector $v$, it is known 
that the sequence $f_{m}(\Q\L)v$ (see 
equation \eqref{truncated_Faber}) converges to $f(\Q\L)v$  
for any analytic  function $f(z)$ defined on $\Omega$, 
provided the spectrum of $\Q\L$ is 
in $\Omega$ (see \cite{moret2001computation}). 
Moreover, by the properties of Faber polynomials,
it is known that the sequence $f_{m}(\Q\L)$ approximates 
asymptotically $f(\Q\L)$ on $\Omega$, as well as
the sequence of best uniform approximation polynomials. 
In this sense, $f_{m}(\Q\L)$ is said to 
be {\em asymptotically optimal} \cite{Eiermann}.
In particular, if we consider the exponential function 
$f(z)=e^{tz}$ and the conformal map \eqref{circle_ellipse}, 
this yields the following $m$-th order Faber 
approximation of the orthogonal dynamics semigroup
\begin{equation}
e^{t\Q\L} \simeq \sum_{j=0}^m  
\frac{1}{(-c_1)^{j/2}} e^{tc_0}J_j(2t\sqrt{-c_1})
\F_j(\Q\L).
\end{equation}

\bibliographystyle{plain}
\bibliography{Faber}

\end{document}